\documentclass[11pt]{article}
\usepackage{subfig,amsmath,amssymb,amsthm,graphicx,tikz,bm,geometry,setspace}
\usepackage[authoryear,longnamesfirst]{natbib}
\newtheorem{lemma}{Lemma}
\newtheorem{theorem}{Theorem}
\newtheorem{proposition}{Proposition}
\newtheorem{corollary}{Corollary}

\newtheorem{assumption}{Assumption}
\newtheorem{remark}{Remark}
\newtheorem{algorithm}{Algorithm}

\title{Inference based on Kotlarski's Identity\thanks{The first arXiv date: August 28, 2018. We would like to thank numerous scholars, seminar participants at Aarhus University, Duke University, KU Leuven, Tinbergen Institute, Universidad Carlos III de Madrid, University of California Davis, University of California Los Angeles, University of California San Diego, University of Illinois Urbana Champtaign, University of Oxford, University of Texas Austin, and University of Wisconsin Madison, and conference participants at AMES 2019, NASM 2019, and New York Camp Econometrics XVI for very helpful comments. The usual disclaimer applies.}
}
\author{
Kengo Kato\thanks{K. Kato: Department of Statistics and Data Science, Cornell University, 1194 Comstock Hall, Ithaca, NY 14853. Email: kk976@cornell.edu}
\and 
Yuya Sasaki\thanks{Y. Sasaki: Department of Economics, Vanderbilt University, VU Station B \#351819, 2301 Vanderbilt Place, Nashville, TN 37235-1819. Email: yuya.sasaki@vanderbilt.edu}
\and 
Takuya Ura\thanks{T. Ura: Department of Economics, University of California, Davis, 1151 Social Sciences and Humanities, Davis, CA 95616. Email: takura@ucdavis.edu}
}
\date{}
\begin{document}
\maketitle
\begin{abstract}
Kotlarski's identity has been widely used in applied economic research.
However, how to conduct inference based on this popular identification approach has been an open question for two decades.
This paper addresses this open problem by constructing a novel confidence band for the density function of a latent variable in repeated measurement error model. 
The confidence band builds on our finding that we can rewrite Kotlarski's identity as a system of linear moment restrictions. 
The confidence band controls the asymptotic size uniformly over a class of data generating processes, and it is consistent against all fixed alternatives. Simulation studies support our theoretical results.

\begin{description}\item {\bf Keywords:} deconvolution, Kotlarski's identity, measurement error, uniform confidence band.
\end{description}
\end{abstract}

\section{Introduction}
Empirical researchers are often interested in recovering features of unobserved variables in economic models.
Kotlarski's identity \citep*{kotlarski1967} -- see also \citet{rao:1992} -- is one of the most popular tools used to identify probability density functions of unobserved latent variables.
Since its first introduction to econometrics by \citet*{li/vuong:1998}, Kotlarski's identity has been widely used in economics when data admit repeated measurements.
Examples of research topics that use Kotlarski's identity include, but are not limited to, empirical auctions \citep[e.g.,][]{LiPerrigneVuong2000,krasnokutskaya:2011}, income dynamics \citep[e.g.,][]{bonhomme/robin:2010}, and labor economics \citep[e.g.,][]{cunha/heckman/navarro:2005,cunha/heckman/schennach:2010,bonhomme/sauder:2011,kennan/walker:2011}.
In these applications, researchers are interested in identifying the probability density function $f_X$ of a latent variable $X$ among others.
The variable $X$ of interest is not observed in data, but two measurements $(Y_1,Y_2)$ are available in data with classical errors, $U_1 = Y_1-X$ and $U_2 = Y_2-X$.
Kotlarski's identity is a nonparametric identifying restriction for the probability density function $f_X$ of $X$ implied by this setup.

The existing econometric literature on Kotlarski's identity focuses on identification and consistent estimation of $f_X$ and related objects \citep[e.g.,][]{li/vuong:1998,Li2002,schennach:2004,Schennach2004et,Schennach2008,bonhomme/robin:2010,Evdokimov2010,ZindeWalsh2014,SongSchennachWhite2015,FirpoGalvaoSong2017} -- also see surveys on this literature by \citet{Chen/Hong/Nekipelov:2011} and \citet{Schennach:2016}.
On the other hand, satisfactory inference methods for $f_X$ are missing in this literature -- in fact, even the \textit{sharp} rate of convergence is unknown for the estimators based on Kotlarski's identity under \textit{unrestrictive} assumptions, and hence a limit distribution result is unavailable under such assumptions. 
Indeed some empirical papers implement nonparametric bootstrap without a theoretical guarantee.
In light of the demands by the empirical researchers for an inference method, and given the current unavailability of theoretically supported methods of inference, we propose a method of inference based on Kotlarski's identity in this paper.

This paper proposes an inference method based on Kotlarski's identity by developing a confidence band for $f_X$.
Our construction of confidence bands works as follows.
First, we derive linear complex-valued moment restrictions based on Kotlarski's identity.
Second, we let the Hermite polynomial sieve \citep[cf.][]{chen:2007} approximate unknown probability density functions.
Third, for a given sieve dimension and for a given class of probability density functions, we compute a bias bound for the linear complex-valued moment restrictions, and slack the linear complex-valued moment restrictions by this bias bound.
Fourth, applying \citet*{chernozhukov/chetverikov/kato:2017}, we compute the uniform norm of the self-normalized process of the slacked linear complex-valued moment restrictions as the test statistic for each point in a set of sieve coefficients.
Fifth, inverting this test statistic in the spirit of \cite{anderson/rubin:1949} yields a confidence set of sieve approximations to possible probability density functions.
Sixth, for a given sieve dimension and for a given class for probability density functions, we compute a bias bound for sieve approximations of probability density functions, and the desired confidence band is obtained by uniformly enlarging the set of sieve approximations by this bias bound.

The process of identifying $f_X$ in additive measurement error models is called \textit{deconvolution} -- for solving convolution integral equations.
There are a number of existing papers on nonparametric inference in deconvolution.
\citet{bissantz/lutz/holzmann/munk:2007}, \citet{bissantz/holtzmann:2008}, \citet{es/gugushvili:2008}, \citet{lounici2011}, and \citet{schmidt-hieber2013} develop uniform confidence bands for $f_X$ under the assumption of known error distributions.\footnote{These paper are based on the literature on deconvolution under known error distribution \citep[e.g.,][]{carroll/hall:1988,stefanski/carroll:1990,fan:1991b,carrasco2011spectral}. \cite{fan:1991a} develops a point-wise asymptotic inference result in this framework.}
In most economic applications, however, it is not plausible to assume that the error distributions are known.
More recently, \citet{kato/sasaki:2018} and \citet{adusumilli/otsu/whang:2017} develop uniform confidence bands for $f_X$ and the distribution function, respectively, without assuming that the error distributions are known, but they both assume that at least one error distribution is symmetric.\footnote{These paper are based on the literature on deconvolution under unknown error distribution with auxiliary data or symmetric error distributions \citep[e.g.,][]{diggle/hall:1993,horowitz/markatou:1996,neumann:1997,efromovich:1997,delaigle/hall/meister:2008,johannes:2009,comte/lacour:2011,delaigle/hall:2015}.}
Kotlarski's identity is a powerful device for new identification results which require neither the known error distribution assumption nor the symmetric error distribution assumption.
This useful feature attracts many economic applications including those listed above, but no econometrician has developed a method of inference in this framework for twenty years ever since its first introduction by \citet*{li/vuong:1998} until our present paper.

It it not surprising that such an inference method has been missing for long in the literature, given the technical difficulties of the problem.
Deconvolution is an ill-posed inverse problem, and inference under this problem is known to be challenging -- see \citet{bissantz/lutz/holzmann/munk:2007,bissantz/holtzmann:2008,lounici2011,horowitz/lee:2012,hall/horowitz:2013,schmidt-hieber2013,adusumilli/otsu/whang:2017,kato/sasaki:2017b,kato/sasaki:2018,babii:2018,chen/christensen:2018} for existing papers developing confidence bands in ill-posed inverse problems for example.
We take a robust inference approach \`a la \citet{anderson/rubin:1949}, and directly work with the moment restrictions based on Kotlarski's identity.
A positive side product of taking this approach is that we do not need to assume the non-vanishing characteristic functions (i.e., we do not need the completeness), which is commonly assumed for nonparametric identification or inversion.


It is also worth mentioning that we chose to use the Hermite polynomial sieve among other sieves in this paper.
The Hermite polynomial sieve has been in fact already known in the literature to be useful to approximate ``smooth density with unbounded support''  \citep{chen:2007} -- also see her discussion of \cite{gallant/nychka:1987} therein.
In addition to this known advantage, we also find this sieve particularly useful for the deconvolution problem.
Note that the deconvolution problem involves applications of the Fourier transform operation and the inverse Fourier transform operation.
To our convenience, the Hermite functions are eigen-functions of the Fourier transform operator.
While we deal with simultaneous restrictions in terms of density and characteristic functions, we can use the Hermite polynomial sieve to approximate both the density and characteristic functions without having to apply the Fourier transform or the Fourier inverse because of the eigen-function property.
This convenient property saves computational time and resources as costly numerical integration within each iteration of a numerical optimization routine would be necessary if any other sieve were used. 
Furthermore, we find that a couple of properties of the Hermite functions (namely the Schr\"odinger equation for a harmonic oscillator and a pair of recursive equations) can be exploited to obtain informative bias bounds of the Hermite polynomial sieve, which in turn contributes to informative inference we establish in this paper.


The rest of the paper is organized as follows.
Section \ref{sec:moment_restrictions} derives linear complex-valued moment restrictions based on Kotlarski's identity.
Section \ref{sec:construction} presents how to compute the confidence band.
Section \ref{sec:properties} presents asymptotic properties of the confidence band.
Section \ref{sec:practical} discusses practical considerations.
Section \ref{sec:simulation} illustrates simulation studies.
The paper concludes in Section \ref{sec:conclusion}.
All mathematical derivations and details are delegated to the appendix.

\section{Linear Complex-Valued Moment Restrictions}\label{sec:moment_restrictions}
Consider the repeated measurement model
 \begin{equation}
 \label{eq:main}
\begin{cases}
&Y_1 = X + U_1, \\
&Y_2 = X + U_2,
\end{cases}
\end{equation}
where $Y_1$ and $Y_2$ are observed, but none of $X$, $U_1$, or $U_2$ is observed. 
We are interested in making inference on the probability density function $f_X$ of $X$. 
We equip this model with the following assumption.

\begin{assumption}\label{a:continuous_independence}${}$
\begin{enumerate}
\item[(i)] 
$X, \ U_1$, and $U_2$ are continuous random variables with finite first moments, and $U_1$ has mean zero.  
\item[(ii)]
$X, \ U_1$, and $U_2$ are mutually independent. 
\end{enumerate}
\end{assumption}

This assumption is standard in the literature on identification and estimation based on Kotlarski's identity \citep[e.g.,][]{li/vuong:1998}.
In fact, the existing literature imposes an additional assumption, namely the identification condition (non-vanishing characteristic function or the completeness) -- see Lemma \ref{lemma:kotlarski} ahead for a specific condition.
We do not invoke such an identification assumption for the purpose of identification-roust inference -- see Remark \ref{remark:non_vanishing} ahead for further details.

We now fix basic notations.
In what follows, $\mathbb{E}_P$ and $\mathbb{V}_P$ denote the expectation and variance operators, respectively, with respect to a joint distribution $P$ of $(Y_1,Y_2)$.
Analogously, $\mathbb{E}_n$ and $\mathbb{V}_n$ denote the expectation and variance operators, respectively, with respect the empirical distribution of $n$ independent copies of $(Y_1,Y_2)$.
We let $i=\sqrt{-1}$ denote the imaginary unit.
For the set of absolutely integrable functions, $\mathcal{L}^1$, we define the Fourier transform $\mathcal{F}$ on $\mathcal{L}^{1}$ by $[\mathcal{F}f](t)=\int_{-\infty}^{\infty} e^{itx}f(x)dx$, and its inverse transform is $[\mathcal{F}^{-1}\phi](x)=\frac{1}{2\pi}\int_{-\infty}^{\infty} e^{-itx}\phi(t)dt$ -- see \citet{folland2007}.
In light of Assumption \ref{a:continuous_independence} (i), we let $f_X$, $f_{U_1}$, and $f_{U_2}$ denote the density functions of $X$, $U_1$, and $U_2$, respectively.
Further, we denote the characteristic functions of them by $\phi_X=\mathcal{F}f_X$, $\phi_{U_1}=\mathcal{F}f_{U_1}$, and $\phi_{U_2}=\mathcal{F}f_{U_2}$.
We first review the existing result of the identification.

\begin{lemma}[Kotlarski's Identity]\label{lemma:kotlarski}
For every joint distribution $P$ of $(Y_1,Y_2)$ satisfying Assumption \ref{a:continuous_independence} for (\ref{eq:main}) and 
$\mathbb{E}_{P}\left[e^{itY_2}\right] \neq 0$ for all $t \in \mathbb{R}$, 
\begin{align}\label{eq:kotlarski_identity}
\phi_X(t) = \exp\left( \int_0^t \frac{i \mathbb{E}_P\left[Y_1 e^{i\tau Y_2}\right]}{\mathbb{E}_P\left[e^{i\tau Y_2}\right]} d\tau \right).
\end{align}
\end{lemma}

This lemma presents Kotlarski's identity due to \cite*{kotlarski1967} -- see also \citet{rao:1992}.
Since it is stated as a lemma, Kotlarski's identity is also known as Kotlarski's lemma or the lemma of Kotlarski in the econometrics literature.
\citet{li/vuong:1998} first introduced it into econometrics and statistics, followed by a series of extensions \citep{Li2002,schennach:2004,bonhomme/robin:2010,Evdokimov2010}.
Some of these extensions relax the assumptions for identification and estimation in various ways.
We do not need to rely on the prototypical assumptions for our purpose of inference, even though they are stated in Lemma \ref{lemma:kotlarski} for convenience of a concise review.
Lemma \ref{lemma:kotlarski} shows that the characteristic function $\phi_X$ of $X$ is explicitly identified by the joint distribution of $(Y_1,Y_2)$.
Under the additional assumption of absolutely integrable characteristic function $\phi_X$, the formula $f_X=\mathcal{F}^{-1}\phi_X$ in turn yields the identification of the probability density function $f_X$ of $X$.

Uniform convergence rates for the estimator of $f_X$ based on Kotlarski's identity are discovered in the existing literature \citep{li/vuong:1998,Li2002,schennach:2004,bonhomme/robin:2010,Evdokimov2010}, but the \textit{sharp} rates under \textit{unrestrictive} assumptions are still unknown.
In particular, limit distribution results under such assumptions are still unknown in the existing literature.
This paper does not aim to derive a non-degenerate limit distribution for any estimator, but it aims to conduct an inference on $f_X$. 
With this said, our proposed inference does not rely on an explicit identifying formula.
We argue that rewriting Kotlarski's lemma in terms of moment restrictions suffices and serves even more conveniently for the sake of conducting inference.

\begin{theorem}[Linear Complex-Valued Moment Restrictions]\label{theorem:moment_condition}
For every joint distribution $P$ of $(Y_1,Y_2)$ satisfying Assumption \ref{a:continuous_independence} for (\ref{eq:main}), \begin{equation}\label{eq:moment_condition}
\mathbb{E}_P\left[\left(iY_1\phi_X(t)-\phi_X^{(1)}(t)\right)\exp(itY_2)\right]=0
\end{equation}
holds for every real $t$, where $\mathbb{E}_P$ is the expectation under $P$.
\end{theorem}

A proof is provided in Appendix \ref{sec:theorem:moment_condition}.

\begin{remark}\label{remark:non_vanishing}
Taking a few more steps beyond the claim in Theorem \ref{theorem:moment_condition} will lead us to the identification result of Lemma \ref{lemma:kotlarski} under the additional assumption of the invertibility or non-vanishing characteristic functions (also known as the completeness) -- see \citet{dhaultfoeuille:2011}.\footnote{\citet{evdokimov2012some} provides a relaxed assumption for the identification.}
For the purpose of inference, however, it is not essential to solve the inverse problem, and thus we stop short of obtaining the explicit formula (\ref{eq:kotlarski_identity}), and only use the moment condition (\ref{eq:moment_condition}).
This idea is analogous to that of \citet{santos2011instrumental,santos2012inference}, where robust inference for functional parameters is conducted without assuming the completeness.
\end{remark}

\section{Construction of the Confidence Band}\label{sec:construction}
Our objective is to construct a confidence band for the probability density function $f_X$ of $X$ on an interval $I \subset \mathbb{R}$. 
The construction procedure is based on the linear complex-valued moment restriction (\ref{eq:moment_condition}).
Throughout, we focus on the set of probability density functions given by
\begin{equation*}
\mathcal{L} \subset \left\{ f \in \mathcal{L}^1 \cap \mathcal{L}^2 \ : \ f \text{ is a probability density function and } \mathcal{F}f \in \mathcal{L}^1 \right\}.
\end{equation*}
For this set of candidate probability density functions, we use $\mathcal{L}^1$ for applying the Fourier transform and the inverse, whereas $\mathcal{L}^2$ is used to approximate $\mathcal{L}$ by an orthonormal basis $\Psi=\{\psi_j \ : \ j=0,1,\ldots\}$ of $\mathcal{L}^2$ -- see Section \ref{sec:practical} for the example of the Hermite basis.

We use a $(q+1)$-dimensional sieve basis $\{\psi_0,\ldots,\psi_q\} \subset \Psi = \{\psi_j : j =0,1,\ldots\}$, with $\psi_j \in \mathcal{L}^1 \cap \mathcal{L}^2$ and $\mathcal{F} \psi_j \in \mathcal{L}^1$ for each $j \in \mathbb{N}$, to approximate the probability density function $f_X$. 
Let $\Theta^{q+1}\subset \mathbb{R}^{q+1}$ be a compact set, and write $\boldsymbol{\psi}=(\psi_0,\dots,\psi_q)^T$.
With a uniform tolerance level $\eta > 0$, each $f \in \mathcal{L}$ is approximated by 
$
x \mapsto \boldsymbol{\psi}(x)^T \boldsymbol{\theta}
$
for some $\boldsymbol{\theta}=(\theta_0,\ldots,\theta_q)^T \in \Theta^{q+1}$, i.e.,
$
\sup_{x \in I} \left| f(x) - \boldsymbol{\psi}(x)^T\boldsymbol{\theta} \right| \leq \eta.
$
The set of values of the sieve coefficients $\boldsymbol{\theta} \in \Theta^{q+1}$ approximating a probability density function $f \in \mathcal{L}$ in this manner is denoted by
\begin{equation}\label{eq:bq}
\mathbf{B}_{q+1,\eta}(f)=\left\{\boldsymbol{\theta}\in\Theta^{q+1}: \sup_{x\in I}|f(x)-\boldsymbol{\psi}(x)^T\boldsymbol{\theta}|\leq\eta\right\}.
\end{equation}

We next incorporate the linear complex-valued moment restrictions (\ref{eq:moment_condition}) in this sieve framework.
For every function $\psi \in \mathcal{L}$ and for every frequency $t \in \mathbb{R}$, 
define
\begin{align}
{R}_{\psi,t}(y_1,y_2)&=-\cos(ty_2)(y_1\mathrm{Im}({\phi}(t))+\mathrm{Re}({\phi}^{(1)}(t)))-\sin(ty_2)(y_1\mathrm{Re}({\phi}(t))-\mathrm{Im}({\phi}^{(1)}(t)))
&&\text{and}
\label{eq:def_R}
\\
{I}_{\psi,t}(y_1,y_2)&=\cos(ty_2)(y_1\mathrm{Re}({\phi}(t))+\mathrm{Im}({\phi}^{(1)}(t)))-\sin(ty_2)(y_1\mathrm{Im}({\phi}(t))-\mathrm{Re}({\phi}^{(1)}(t))),
&&
\label{eq:def_I}
\end{align}
where $\phi = \mathcal{F} \psi$. 
Note that $\mathrm{Re}(\cdot)$ (resp., $\mathrm{Im}(\cdot)$) denotes the real (resp., imaginary) part of a complex number. 
Further, stack these functions across $\psi \in \{\psi_0,\dots,\psi_q\}$ to define the random vector  
\begin{align*}
\mathbf{R}_{t}&=({R}_{\psi_0,t}(Y_1,Y_2),\ldots,{R}_{\psi_q,t}(Y_1,Y_2))^T
\qquad\text{and}\\
\mathbf{I}_{t}&=({I}_{\psi_0,t}(Y_1,Y_2),\ldots,{I}_{\psi_q,t}(Y_1,Y_2))^T.
\end{align*}
With these notations, we now represent the linear complex-valued moment restrictions (\ref{eq:moment_condition}) for the sieve approximation by
\begin{align}
|\mathbb{E}_n[\mathbf{R}_{t}]^T\boldsymbol{\theta}| \leq \delta(t)
\qquad\text{and}\qquad
|\mathbb{E}_n[\mathbf{I}_{t}]^T\boldsymbol{\theta}| \leq \delta(t)
\label{eq:moment_inequalities}
\end{align}
for all $t \in [-T,T]$ for $T \in (0,\infty)$, where $\delta(t)>0$ is the tolerance level of sieve approximation error for each $t \in [-T,T]$.

Our construction of the confidence band is based on a test statistic that quantifies the extent of deviation from the moment inequalities (\ref{eq:moment_inequalities}). 
To construct a feasible test statistic, we use a grid $\{t_1,\dots,t_L\} \subset [-T,T]$ of $L$ frequencies.
Define the test statistic by
\begin{equation*}
T(\boldsymbol{\theta})=\sqrt{n}\max_{1\leq l\leq L}\max\left\{\frac{|\mathbb{E}_n[\mathbf{R}_{t_l}]^T\boldsymbol{\theta}|-\delta(t_l)}{\sqrt{\boldsymbol{\theta}^T\mathbb{V}_n(\mathbf{R}_{t_l})\boldsymbol{\theta}}},\frac{|\mathbb{E}_n[\mathbf{I}_{t_l}]^T\boldsymbol{\theta}|-\delta(t_l)}{\sqrt{\boldsymbol{\theta}^T\mathbb{V}_n(\mathbf{I}_{t_l})\boldsymbol{\theta}}}\right\}
\end{equation*}
for each $\boldsymbol{\theta}\in\Theta^{q+1}$.
Let $\alpha\in(0,1/2)$ be fixed.
We define the critical value $c(\alpha,\boldsymbol{\theta})$ of this statistic $T(\boldsymbol{\theta})$ by the conditional $(1-\alpha)$-th quantile of the multiplier bootstrap statistic
$$
\sqrt{n}\max_{1\leq l\leq L}\max\left\{\frac{\left|\mathbb{E}_n[\boldsymbol{\epsilon}(\mathbf{R}_{t_l}-\mathbb{E}_n[\mathbf{R}_{t_l}])]^T\boldsymbol{\theta}\right|}{\sqrt{\boldsymbol{\theta}^T\mathbb{V}_n(\mathbf{R}_{t_l})\boldsymbol{\theta}}},\frac{\left|\mathbb{E}_n[\boldsymbol{\epsilon}(\mathbf{I}_{t_l}-\mathbb{E}_n[\mathbf{I}_{t_l}])]^T\boldsymbol{\theta}\right|}{\sqrt{\boldsymbol{\theta}^T\mathbb{V}_n(\mathbf{I}_{t_l})\boldsymbol{\theta}}} \right\}
$$
given the data, where $\boldsymbol{\epsilon}_1,\ldots,\boldsymbol{\epsilon}_n$ are independent standard normal random variables independent of the data. 
As a more conservative yet simpler alternative following \citet[][eq. (19)]{chernozhukov/chetverikov/kato:2017}, we may define the critical value as
\begin{equation*}\label{eq:critical_value}
c(\alpha) = \frac{\Phi^{-1}(1-\alpha/(4L))}{1-\Phi^{-1}(1-\alpha/(4L))^2/n},
\end{equation*}
where $\Phi$ is the cumulative distribution function of the standard normal distribution.
Our confidence band for the density function of $X$ is given by the $I$-restriction of
\begin{equation}\label{eq:confidence_band}
\mathcal{C}_{n}(\alpha)=\left\{f\in\mathcal{L}:\ T(\boldsymbol{\theta})\leq c(\alpha,\boldsymbol{\theta})\mbox{ for some }\boldsymbol{\theta}\in\mathbf{B}_{q+1,\eta}(f)\right\},
\end{equation}
where $\mathbf{B}_{q+1,\eta}(f)$ is defined in (\ref{eq:bq}).

A practical procedure to obtain this confidence band is outlined as Algorithm \ref{algorithm:main} below.
While this algorithm prescribes a general procedure, we also elaborate on details of practical considerations in Section \ref{sec:practical}, where we introduce the Hermite orthonormal sieve (Section \ref{sec:sieve}), propose concrete choice rules for the tuning parameters (Section \ref{sec:tuning_parameters}), and present a more concrete algorithm for these sieve and tuning parameters (Section \ref{sec:implementation}).

\begin{algorithm}\label{algorithm:main}${}$
\begin{enumerate}
\item For each $x \in I$, compute 
\begin{align*}
f^L(x) = \min_{\boldsymbol{\theta} \in \Theta^{q+1}} \boldsymbol{\psi}(x)^T \boldsymbol{\theta}
 \quad\text{subject to}\quad & T(\boldsymbol{\theta}) \leq c(\alpha,\boldsymbol{\theta})
\\
& \boldsymbol{\psi}(x)^T \boldsymbol{\theta} \geq -\eta \text{ for all } x \in I
\\
& \left| \sqrt{2\pi} \boldsymbol{\psi}(0)^T \text{diag}\left(1,0,-1,0\dots\right) \boldsymbol{\theta} -1 \right| \leq \eta
\end{align*}
\item For each $x \in I$, compute 
\begin{align*}
f^U(x) = \max_{\boldsymbol{\theta} \in \Theta^{q+1}} \boldsymbol{\psi}(x)^T \boldsymbol{\theta}
 \quad\text{subject to}\quad & T(\boldsymbol{\theta}) \leq c(\alpha,\boldsymbol{\theta})
\\
& \boldsymbol{\psi}(x)^T \boldsymbol{\theta} \geq -\eta \text{ for all } x \in I
\\
& \left| \sqrt{2\pi} \boldsymbol{\psi}(0)^T \text{diag}\left(1,0,-1,0\dots\right) \boldsymbol{\theta} -1 \right| \leq \eta
\end{align*}
\item The confidence band is set to $\left[ f^L(x) - \eta, \ f^U(x) + \eta \right]$, $x \in I$.
\end{enumerate}
\end{algorithm}

\begin{remark}
The baseline idea of our confidence band construction is to discretize the continuum of moment conditions (\ref{eq:moment_condition}) into a set of finite but many moment inequalities on the sieve coefficients $\boldsymbol{\theta}$, calibrate critical values by the multiplier bootstrap for the max statistic as in \cite{chernozhukov/chetverikov/kato:2017}, and project a confidence set for $\boldsymbol{\theta}$ into a confidence band for $f_{X}$.  
A natural question would be whether we could use directly the continuum of moment conditions (\ref{eq:moment_condition}) without discretization, similarly to, e.g., \cite{AndrewsShi2013,ANDREWS2017275}. One difficulty, however, is that the moment functions corresponding to the moment condition (\ref{eq:moment_condition}) at given $f_{X}$, i.e., $\{ (y_{1},y_{2}) \mapsto R_{f_{X},t}(y_{1},y_{2}), I_{f_{X},t}(y_{1},y_{2}) : t \in \mathbb{R} \}$, is not likely to be Donsker, in view of the fact that, e.g.,  the function class $\{  y_{2} \mapsto \cos (ty_{2}) : t \in \mathbb{R} \}$ is non-Donsker as soon as $Y_{2}$ has a density (which is the case under our assumption), so that the ``manageability'' condition in \cite{AndrewsShi2013,ANDREWS2017275} would not be satisfied in our case.  Indeed, the preceding function class is not Glivenko-Cantelli from the Riemann-Lebesgue lemma and discreteness of the empirical distribution; see \cite{FeuervergerMureika1977} for details. Another potential approach would be to apply the method of continuum of moment conditions developed in \cite{CarrascoFlorens2000}. Their analysis relies on point-identification of the parameter of interest and, more importantly, focused on a finite dimensional parameter of interest (so the convergence rate of their estimator is the parametric rate), so that their approach is not directly applicable to our problem.
\end{remark}

\section{Properties of the Confidence Band}\label{sec:properties}
In this section, we present theoretical properties of the confidence band (\ref{eq:confidence_band}).
Let $\mathcal{P}$ denote a given space to which the joint distribution of $(Y_1,Y_2)$ belongs.
For every $P\in\mathcal{P}$, define the identified set
\begin{equation}\label{eq:L0}
\mathcal{L}_0(P) = \left\{ f\in\mathcal{L} \ : \ \phi = \mathcal{F}f \ \text{and} \ \mathbb{E}_P\left[\left(iY_1\phi(t)-\phi^{(1)}(t)\right)\exp(itY_2)\right]=0 \ \text{for every } t \in \mathbb{R} \right\}
\end{equation}
as the set of density functions for which the linear complex-valued moment restriction (\ref{eq:moment_condition}) is satisfied.
Furthermore, we define the sieve-approximation counterpart of $\mathcal{L}_0(P)$ by
\begin{equation}\label{eq:L0star}
\mathcal{L}_0^\ast(P)=
\left\{
f\in\mathcal{L}:
\inf_{\boldsymbol{\theta}_\ast\in\mathbf{B}_{q+1,\eta}(f)}\max_{1\leq l\leq L}\left(\max\{|\mathbb{E}_P[\mathbf{R}_{t_l}]^T\boldsymbol{\theta}_\ast|,|\mathbb{E}_P[\mathbf{I}_{t_l}]^T\boldsymbol{\theta}_\ast|\}-\delta(t_l)\right) \leq 0
\right\},
\end{equation}
where $\mathbf{B}_{q+1,\eta}(f)$ is defined in (\ref{eq:bq}).
Here, the infimum over the empty set is understood to be the infinity. 

The current section is structured as follows.
First, we establish the size control for the confidence band (\ref{eq:confidence_band}) to contain the approximation set $\mathcal{L}_0^\ast(P)$ in Section \ref{sec:size}.
Second, we establish the containment of the identified set by the approximation set (i.e., $\mathcal{L}_0(P) \subset \mathcal{L}_0^\ast(P)$) in Section \ref{sec:bounds}.
These two pieces of the results together show the validity of the confidence band (\ref{eq:confidence_band}) to contain the identified set $\mathcal{L}_0(P)$.
Lastly, Section \ref{sec:power} presents power properties with local alternatives.
Throughout, we assume to observe $n$ i.i.d. copies of $(Y_1,Y_2)$ drawn from $P \in \mathcal{P}$.

\subsection{Size Control}\label{sec:size}
We make the following assumption for a uniform size control. 
\begin{assumption}\label{a:size_control}
(i) $\mathbb{E}_P\left[Y_1^2\right]<\infty$.
(ii) There are constants $0 < c_1<1/2$ and $C_1 > 0$ such that 
$$
\left(M_{L,q,3}^3(\boldsymbol{\theta},P) \vee M_{L,q,4}^2(\boldsymbol{\theta},P) \vee B_{L,q}(\boldsymbol{\theta},P) \right)^2 \log^{7/2}(4Ln)\leq C_1n^{1/2-c_1}
$$
for all $P \in \mathcal{P}$ and $\theta \in \Theta^{q+1}$, where\\{\small
$
M_{L,q,k}(\boldsymbol{\theta},P) = 
\max_{1 \leq l \leq L} \max
\left\{\mathbb{E}_P\left[\left|\frac{(\mathbf{R}_{t_l}-\mathbb{E}_P[\mathbf{R}_{t_l}])^T\boldsymbol{\theta}}{\sqrt{\boldsymbol{\theta}^T\mathbb{V}_P(\mathbf{R}_{t_l})\boldsymbol{\theta}}}\right|^k\right]^{1/k}, \mathbb{E}_P\left[\left|\frac{(\mathbf{I}_{t_l}-\mathbb{E}_P[\mathbf{I}_{t_l}])^T\boldsymbol{\theta}}{\sqrt{\boldsymbol{\theta}^T\mathbb{V}_P(\mathbf{I}_{t_l})\boldsymbol{\theta}}}\right|^k\right]^{1/k} \right\}
$}\\
and {\small
$
B_{L,q}(\boldsymbol{\theta},P) = 
\mathbb{E}_P\left[\max_{1 \leq l \leq L} \max
\left\{\left|\frac{(\mathbf{R}_{t_l}-\mathbb{E}_P[\mathbf{R}_{t_l}])^T\boldsymbol{\theta}}{\sqrt{\boldsymbol{\theta}^T\mathbb{V}_P(\mathbf{R}_{t_l})\boldsymbol{\theta}}}\right|^4 , \left|\frac{(\mathbf{I}_{t_l}-\mathbb{E}_P[\mathbf{I}_{t_l}])^T\boldsymbol{\theta}}{\sqrt{\boldsymbol{\theta}^T\mathbb{V}_P(\mathbf{I}_{t_l})\boldsymbol{\theta}}}\right|^4 \right\} \right]^{1/4}
$}.
\end{assumption}

\begin{remark}
When $\mathbb{E}_P\left[Y_1^4\right]<\infty$ and $\{\psi_0,\ldots,\psi_q\}$ are the Hermite functions, Assumption \ref{a:size_control} (ii) holds if 
\begin{equation}\label{eq:assn2_sufficient}
(q+1)/
\min_{t\in[-T,T]}\min\left\{
\mathrm{eig}_{\min}(\mathbb{V}_P(\mathbf{R}_{t})),
\mathrm{eig}_{\min}(\mathbb{V}_P(\mathbf{I}_{t}))
\right\}=O(n^{1/6-c_1/3}\log^{-7/6}(4Ln)).
\end{equation}
A derivation of this condition is in Appendix \ref{appendix:proof_assn2_sufficient}.
Eq. (\ref{eq:assn2_sufficient}) restricts how fast  $q$ and $T$ can increase to infinity.\footnote{Loosely speaking, the term $(q+1)/
\min_{t\in[-T,T]}\min\left\{
\mathrm{eig}_{\min}(\mathbb{V}_P(\mathbf{R}_{t})),
\mathrm{eig}_{\min}(\mathbb{V}_P(\mathbf{I}_{t}))
\right\}$ is increasing in $q$ and $T$. When $q$ increases, the numerator increases and the minimum eigenvalues of the $(q+1)$ square matrices, $\mathbb{V}_P(\mathbf{R}_{t})$ and $\mathbb{V}_P(\mathbf{I}_{t})$, can be smaller. Moreover, when $T$ increases, the denominator decreases.}
On the other hand, Eq. (\ref{eq:assn2_sufficient}) does not restrict the choice of $L$ in the sense that, as long as $L$ grows at a polynomial rate of $n$, the term $\log^{-7/6}(4Ln)$ is negligible on the right hand side. 
\end{remark}

\begin{theorem}[Size Control]\label{theorem:size_control}
Suppose that Assumption \ref{a:size_control} holds.
Then, there exist positive constants $c$ and $C$ depending only on $c_{1}$ and $C_{1}$ such that 
$$
\inf_{P\in\mathcal{P}}\inf_{f\in\mathcal{L}_0^\ast(P)}\mathbb{P}_P(f\in \mathcal{C}_{n}(\alpha))\geq 1-\alpha-Cn^{-c}.
$$
\end{theorem}

A proof is provided in Appenedix \ref{sec:theorem:size_control}.
This theorem guarantees the size control for the event of the inclusion of the density function in the approximate identified set $\mathcal{L}_0^\ast(P)$ as opposed to the identified set $\mathcal{L}_0(P)$.
The next section presents conditions under which the approximate identified set $\mathcal{L}_0^\ast(P)$ contains the identified set $\mathcal{L}_0(P)$ for every possible joint distribution $P \in \mathcal{P}$.

\subsection{Bounds of Approximation Errors}\label{sec:bounds}
Throughout, we equip $\mathcal{L}^2$ with the inner product $\langle{\cdot,\cdot}\rangle$ defined by
$$
\langle{f_1,f_2}\rangle = \int_\mathbb{R} f_1(x) f_2(x) dx.
$$

\begin{assumption}\label{a:orthonormal}${}$\\
(i)
$\Psi = \{\psi_j \ : \ j=0,1,\ldots\}$ is an orthonormal basis of $\mathcal{L}^2$, i.e., orthonormal and complete in $\mathcal{L}^2$.
\\
(ii) $\left( \langle{f,\psi_0}\rangle, \dots, \langle{f,\psi_q}\rangle\right) \in \Theta^{q+1}$ for all $f \in \mathcal{L}_0(P)$ for all $P \in \mathcal{P}$.
\end{assumption}

In Section \ref{sec:sieve}, we propose a concrete orthonormal basis $\Psi$ and the set $\Theta^{q+1}$ of coefficients to satisfy Assumption \ref{a:orthonormal}.
The following theorem provides a guide for choices of $\eta$ and $\delta(t_1),\dots,\delta(t_L)$ such that $\mathcal{L}_0^\ast(P)$ contains $\mathcal{L}_0(P)$ for every possible $P \in \mathcal{P}$.

\begin{theorem}[Approximation]\label{theorem:approximation}
Suppose that Assumption \ref{a:orthonormal} is satisfied.
If
\begin{align}
&\sup_{f\in\mathcal{L}}\sup_{t\in I}\left|\sum_{j=q+1}^{\infty}\langle{ f,\psi_j }\rangle \cdot \psi_j(t)\right| \leq \eta 
\qquad\text{and}
\label{eq:approximation_psi}
\\
&\sup_{f\in\mathcal{L}}\left|\sum_{j=q+1}^{\infty} \langle{ f,\psi_j }\rangle \cdot 
\left(
i\phi_j(t_l)\cdot\mathbb{E}_P\left[Y_1\exp(it_lY_2)\right]
-\phi_j^{(1)}(t_l)\cdot\mathbb{E}_P\left[\exp(it_lY_2)\right]
\right)
\right| \leq \delta(t_l)
\label{eq:approximation_moments}
\end{align}
for every $l \in \{1,\ldots,L\}$, then $\mathcal{L}_0(P)\subset\mathcal{L}_0^\ast(P)$ holds for all $P \in \mathcal{P}$.
\end{theorem}

A proof is provided in Appendix \ref{sec:theorem:approximation}.
In Section \ref{sec:sieve}, we provide concrete evaluations of the left-hand side of (\ref{eq:approximation_psi}) and (\ref{eq:approximation_moments}) under a concrete orthonormal basis $\Psi$, namely the Hermite orthonormal basis.
Putting Theorems \ref{theorem:size_control} and \ref{theorem:approximation} together, we obtain the following result on the validity of the confidence band.
\begin{corollary}[Validity of the Confidence Band]
Suppose that the conditions of Theorems \ref{theorem:size_control} and \ref{theorem:approximation} are satisfied.
Then, there exist positive constants $c$ and $C$ depending only on $c_{1}$ and $C_{1}$ such that 
$$
\inf_{P\in\mathcal{P}}\inf_{f\in\mathcal{L}_0(P)}\mathbb{P}_P(f\in \mathcal{C}_{n}(\alpha))\geq 1-\alpha-Cn^{-c}.
$$
\end{corollary}

\subsection{Power}\label{sec:power}
We introduce the following short-hand notation for the random variable defined as the maximum deviation of the sample variance from the population variance:
$$
B_{V}=
\sup_{\boldsymbol{\theta}\in \mathbf{B}_{q+1,\eta}(f)}\max_{l=1,\ldots,L}\max\left\{\boldsymbol{\theta}^T(\mathbb{V}_n(\mathbf{R}_{t_l})-\mathbb{V}_P(\mathbf{R}_{t_l}))\boldsymbol{\theta},\boldsymbol{\theta}^T(\mathbb{V}_n(\mathbf{I}_{t_l})-\mathbb{V}_P(\mathbf{I}_{t_l}))\boldsymbol{\theta}\right\}.
$$
The following theorem shows a power property of our proposed inference method.

\begin{theorem}[Power]\label{thm_power}
Suppose that Assumption \ref{a:size_control} (i) holds.
For every $P\in\mathcal{P}$, every $f\in\mathcal{L}$, every $\nu>0$, and every $b\in(0,\infty)$,
if there is $t_\ast \in \{t_1,\dots,t_L\}$ such that at least one of the following statements holds:   
\begin{align}
\sqrt{n}\inf_{\boldsymbol{\theta}\in \mathbf{B}_{q+1,\eta}(f)}\frac{\mathbb{E}_P[\mathbf{R}_{t_\ast}]^T\boldsymbol{\theta}-\delta(t_\ast)}{\sqrt{\boldsymbol{\theta}^T\mathbb{V}_P(\mathbf{R}_{t_\ast})\boldsymbol{\theta}+\nu}}
\geq(1+b)\cdot\mathbb{E}_P\left[\sup_{\boldsymbol{\theta}\in\mathbf{B}_{q+1,\eta}(f)}\frac{|\mathbb{G}_n[\mathbf{R}_{t_\ast}]^T\boldsymbol{\theta}|}{\sqrt{\boldsymbol{\theta}^T\mathbb{V}_n(\mathbf{R}_{t_\ast})\boldsymbol{\theta}}}\right]\notag\\
+\sqrt{2\log(4L)}+\sqrt{2\log(1/\alpha)}
\label{local_alternatves1}
\end{align}
\begin{align}
\sqrt{n}\inf_{\boldsymbol{\theta}\in \mathbf{B}_{q+1,\eta}(f)}\frac{-\mathbb{E}_P[\mathbf{R}_{t_\ast}]^T\boldsymbol{\theta}-\delta(t_\ast)}{\sqrt{\boldsymbol{\theta}^T\mathbb{V}_P(\mathbf{R}_{t_\ast})\boldsymbol{\theta}+\nu}}
\geq(1+b)\cdot\mathbb{E}_P\left[\sup_{\boldsymbol{\theta}\in\mathbf{B}_{q+1,\eta}(f)}\frac{|\mathbb{G}_n[\mathbf{R}_{t_\ast}]^T\boldsymbol{\theta}|}{\sqrt{\boldsymbol{\theta}^T\mathbb{V}_n(\mathbf{R}_{t_\ast})\boldsymbol{\theta}}}\right]\notag\\
+\sqrt{2\log(4L)}+\sqrt{2\log(1/\alpha)}
\label{local_alternatves2}
\end{align}
\begin{align}
\sqrt{n}\inf_{\boldsymbol{\theta}\in \mathbf{B}_{q+1,\eta}(f)}\frac{\mathbb{E}_P[\mathbf{I}_{t_\ast}]^T\boldsymbol{\theta}-\delta(t_\ast)}{\sqrt{\boldsymbol{\theta}^T\mathbb{V}_P(\mathbf{I}_{t_\ast})\boldsymbol{\theta}+\nu}}
\geq(1+b)\cdot\mathbb{E}_P\left[\sup_{\boldsymbol{\theta}\in\mathbf{B}_{q+1,\eta}(f)}\frac{|\mathbb{G}_n[\mathbf{I}_{t_\ast}]^T\boldsymbol{\theta}|}{\sqrt{\boldsymbol{\theta}^T\mathbb{V}_n(\mathbf{I}_{t_\ast})\boldsymbol{\theta}}}\right]\notag\\
+\sqrt{2\log(4L)}+\sqrt{2\log(1/\alpha)}
\label{local_alternatves3}
\end{align}
\begin{align}
\sqrt{n}\inf_{\boldsymbol{\theta}\in \mathbf{B}_{q+1,\eta}(f)}\frac{-\mathbb{E}_P[\mathbf{I}_{t_\ast}]^T\boldsymbol{\theta}-\delta(t_\ast)}{\sqrt{\boldsymbol{\theta}^T\mathbb{V}_P(\mathbf{I}_{t_\ast})\boldsymbol{\theta}+\nu}}
\geq(1+b)\cdot\mathbb{E}_P\left[\sup_{\boldsymbol{\theta}\in\mathbf{B}_{q+1,\eta}(f)}\frac{|\mathbb{G}_n[\mathbf{I}_{t_\ast}]^T\boldsymbol{\theta}|}{\sqrt{\boldsymbol{\theta}^T\mathbb{V}_n(\mathbf{I}_{t_\ast})\boldsymbol{\theta}}}\right]\notag\\
+\sqrt{2\log(4L)}+\sqrt{2\log(1/\alpha)},
\label{local_alternatves4}
\end{align}
then 
$$
\mathbb{P}_P(f\notin \mathcal{C}_{n}(\alpha))
\geq
\mathbb{P}_P\left(B_{V}\leq \nu\right)-\frac{1}{1+b}.
$$
\end{theorem}

A proof is provided in Appendix \ref{sec:thm_power}.
According to this theorem, for any density function $f \in \mathcal{L}$ such that at least one of the moment inequalities violated at some frequency point $t_\ast$ in the grid $\{t_1,\dots,t_L\}$, then the probability that this density function does not belong to the confidence band is bounded below by $\mathbb{P}_P\left(B_{V}\leq \nu\right)-\frac{1}{1+b}$.
Choosing sequences of $\nu$ and $b$ so that $\mathbb{P}_P\left(B_{V}\leq \nu\right)-\frac{1}{1+b} \rightarrow 1$ as $n \rightarrow \infty$, therefore, this theorem implies the consistency against all fixed alternatives.

\section{Practical Considerations}\label{sec:practical}
The current section presents a guide to practice.
Construction of the confidence band and its theoretical properties are presented in Sections \ref{sec:construction} and \ref{sec:properties} with abstract objects.
These objects in particular include an orthonormal basis $\Psi = \{\psi_j \ : \ j=0,1,\ldots\}$, the sieve dimension $q$, the set $\Theta^{q+1}$ of sieve coefficients, and the tolerance levels, $\eta$, $\delta(t_1),\dots, \delta(t_L)$, of approximation errors.
Section \ref{sec:sieve} presents concrete choices of $\Psi = \{\psi_j \ : \ j=0,1,\ldots\}$ and $\Theta^{q+1}$.
Section \ref{sec:tuning_parameters} presents a concrete data-driven procedure for selecting the tolerance levels $\eta$, $\delta(t_1),\dots, \delta(t_L)$.
Section \ref{sec:implementation} presents a concrete implementation procedure for constructing the confidence band with these choices of the objects.

\subsection{The Hermite Orthonormal Basis}\label{sec:sieve}
There is a large extent of freedom of choice for an orthonormal basis $\Psi$ -- see \cite{chen:2007} for a list of options.
We recommend the Hermite orthonormal basis in particular for its convenient properties and its nice compatibility with the deconvolution framework -- a Hermite function is an eigenfunction of the Fourier transform and the Fourier inverse.\footnote{The Hermite orthonormal sieve is not location invariant, and hence we recommend to location- and scale normalize the observed data using the empirical moments.}
The Hermite functions take the form
\begin{equation}\label{eq:hermite_function}
\psi_j(x)=\frac{1}{\sqrt{2^{j} j!\sqrt{\pi}}} \cdot \exp(-x^2/2) \cdot H_j(x),
\end{equation}
$j=0,1,\ldots$,
where $H_j$ is the Hermite polynomial defined by
$$
H_j(x)=(-1)^{j} \cdot \exp(x^2) \cdot \frac{d^{j}}{dx^{j}}\exp(-x^2).
$$
The Hermite functions are the eigenfunctions of the Fourier transform operator, and specifically, $\phi_j=\mathcal{F}\psi_j=i^j\sqrt{2\pi}\psi_j$ holds.
For any $q \in \mathbb{N}$ to be chosen below, have the set of sieve coefficients satisfy
\begin{equation}\label{eq:hermite_theta}
\Theta^{q+1} \supset \left[-1.086435\pi^{-1/4},1.086435\pi^{-1/4}\right]^{q+1}.
\end{equation}
These concrete choices are made for the sake of satisfying Assumption \ref{a:orthonormal} so we can use Theorem \ref{theorem:approximation}.

\begin{proposition}[Sufficient Condition for Assumption \ref{a:orthonormal}]\label{prop:hermite_assumption}
If $\Psi=\{\psi_j:j=0,1,\ldots\}$ is the sequence of the Hermite functions given in (\ref{eq:hermite_function}), then it satisfies Assumption \ref{a:orthonormal} with the coefficient set given in (\ref{eq:hermite_theta}).
\end{proposition}

A proof is provided in Appendix \ref{sec:prop:hermite_assumption}.
We also present the following proposition which provides approximation bounds for the condition of Theorem \ref{theorem:approximation} to guarantee the containment of the identified set $\mathcal{L}_0(P)$ by the approximate identified set $\mathcal{L}_0^\ast(P)$ for every possible joint distribution $P \in \mathcal{P}$.

\begin{proposition}[Approximation Bounds]\label{prop:bounds}
Suppose that $\Psi=\{\psi_j:j=0,1,\ldots\}$ is the sequence of the Hermite functions given in (\ref{eq:hermite_function}) and $\Theta^{q+1}$ satisfies (\ref{eq:hermite_theta}).
If
\begin{equation}\label{eq:hermite_smoothness}
\int \left( \frac{d^2}{d x^2} (f^{(2)}(x)+x^2f(x)) + x^2 \cdot (f^{(2)}(x)+x^2f(x)) \right)^2 dx \leq M,
\end{equation}
then, for each $q=1,2,\ldots$, 
\begin{align}
&\sup_{f\in\mathcal{L}}\sup_{x\in I}\left|\sum_{j=q+1}^{\infty}\langle{ f,\psi_j }\rangle \cdot \psi_j(x)\right|
\leq
\frac{1.086435\pi^{-1/4}}{\sqrt{2q+3}}\sqrt{M\sum_{j=q+1}^{\infty}(2j+1)^{-3}},
\label{eq:hermite_smoothness_psi}\\
&\sup_{f\in\mathcal{L}}\sup_{t\in \mathbb{R}}\left|\sum_{j=q+1}^{\infty}\langle{ f,\psi_j }\rangle \cdot \phi_j(t)\right|
\leq
\frac{1.086435\pi^{-1/4}\sqrt{2\pi}}{\sqrt{2q+3}}\sqrt{M\sum_{j=q+1}^{\infty}(2j+1)^{-3}},
\qquad\text{and}
\label{eq:hermite_smoothness_phi}\\
&\sup_{f\in\mathcal{L}}\left|\sum_{j=q+1}^{\infty} \langle{ f,\psi_j }\rangle \cdot \left(i\phi_j(t)\cdot\mathbb{E}_P\left[Y_1\exp(itY_2)\right]-\phi_j^{(1)}(t)\cdot\mathbb{E}_P\left[\exp(itY_2)\right]\right)\right| \notag\\
&\leq
\frac{1.086435\pi^{-1/4}}{\sqrt{2\pi}}\cdot\sqrt{M\sum_{j=q+1}^{\infty}(2j+1)^{-3}}\left(\frac{\mathbb{E}_P\left[|Y_1|\right]}{\sqrt{2q+3}}+1\right)
\qquad\text{for all } t.
\label{eq:hermite_smoothness_moments}
\end{align}
\end{proposition}

A proof is provided in Appendix \ref{sec:prop:bounds}.
An admissible function class in terms of smoothness restriction can be specified by (\ref{eq:hermite_smoothness}).
Note that this is analogous to the standard practice in the literature to work with Sovolev classes of functions.
With the function class specified in this way, equations (\ref{eq:hermite_smoothness_psi}) and (\ref{eq:hermite_smoothness_moments}) prescribe possible choices of the tolerance levels $\eta,\delta(t_1),\dots,\delta(t_L)$ which admit $\mathcal{L}_0(P) \subset \mathcal{L}_0^\ast(P)$ for all $P \in \mathcal{P}$ per Theorem \ref{theorem:approximation}.
See Section \ref{sec:tuning_parameters} for concrete choice procedures.
Equation (\ref{eq:hermite_smoothness_phi}) in addition suggests the worst approximation error for the characteristic function, which will be useful when we impose natural restrictions on the characteristic functions -- see Section \ref{sec:implementation}.
Finally, we present the following proposition showing an alternative representation of the function class specification (\ref{eq:hermite_smoothness}).

\begin{proposition}[Equivalent Smoothness Condition]\label{prop:equivalent_smoothness}
Suppose that $x \mapsto \frac{d^2}{d x^2} (f''(x)+x^2f(x)) + x^2 \cdot (f''(x)+x^2f(x))$ is $\mathcal{L}^2$ and $\lim_{|x| \rightarrow \infty} x^2 f(x) = \lim_{|x| \rightarrow \infty} x^2 f^{(1)}(x) = \lim_{|x| \rightarrow \infty} f^{(2)}(x) = \lim_{|x| \rightarrow \infty} f^{(3)}(x) = 0$. Then (\ref{eq:hermite_smoothness}) is equivalent to
\begin{equation}\label{eq:equivalent_smoothness}
\int \left|\int\left(x^4 - (2t^2 + it)x^2 -(2 - 4it)x + (t^4 + 2)\right)e^{itx}f(x)dx\right|^2 dt \leq M.
\end{equation}
\end{proposition}

A proof is provided in Appendix \ref{sec:prop:equivalent_smoothness}.
The components in the left-hand side of (\ref{eq:equivalent_smoothness}) are estimable by \cite{li/vuong:1998} and its extensions, and thus Proposition \ref{prop:equivalent_smoothness} provides a guideline for setting the smoothness bound parameter $M$ -- see Section \ref{sec:tuning_parameters}.

\subsection{Choice of Tuning Parameters}\label{sec:tuning_parameters}
In this section, we provide example procedures of choosing the smoothness bound $M$ and the tolerance levels $\eta, \delta(t_1),\dots,\delta(t_L) \in (0,\infty)$ in finite sample.
Although we present a data driven choice of the smoothness bound $M$ below, we remark that the smoothness bound $M \in (0,\infty)$ as well as the sieve dimension $q \in \mathbb{N}$ could be imposed by a researcher in the spirit of honest inference \citep[cf.][]{armstrong/kolesar:2018}.


{\bf Smoothness Bound $M$:}
We choose $M$ to satisfy the condition (\ref{eq:hermite_smoothness}) of Proposition \ref{prop:bounds}.
In light of Proposition \ref{prop:equivalent_smoothness}, we choose $M$ to satisfy (\ref{eq:equivalent_smoothness}).
By Minkowski's inequality, it is sufficient to have $M$ satisfy
\begin{align*}
&\left( \int \left|\int x^4 e^{itx} f(x)dx\right|^2 dt \right)^{1/2} + \left( \int \left| 2t^2 + it \right|^2 \left| \int x^2 e^{itx} f(x)dx \right|^2 dt \right)^{1/2} +
\\ 
& \left( \int \left| 2 - 4it \right|^2 \left| \int x e^{itx} f(x)dx \right|^2 dt \right)^{1/2} + \left( \int \left( t^4 + 2 \right)^2 \left| \int e^{itx} f(x)dx \right|^2 dt \right)^{1/2}
\leq M^{1/2}
\end{align*}
Such a bound may be obtained by the plug-in \citep[cf.][]{Schennach2015}.
Specifically, one can choose
\begin{align}
M = 
& \left[ \left( \int_{-T}^{T} \left| \widehat\varphi_X^{(4)}(t) \right|^2 dt \right)^{1/2} + \left( \int_{-T}^{T} \left| 2t^2 + it \right|^2 \left| \widehat\varphi_X^{(2)}(t) \right|^2 dt \right)^{1/2} + \right.
\notag\\ 
& \left. \left( \int_{-T}^{T} \left| 2 - 4it \right|^2 \left| \widehat\varphi_X^{(1)}(t) \right|^2 dt \right)^{1/2} + \left( \int_{-T}^{T} \left( t^4 + 2 \right)^2 \left| \widehat\varphi_X(t) \right|^2 dt \right)^{1/2} \right]^{2},
\label{eq:plug_in_M}
\end{align}
where $\widehat\varphi_X$, $\widehat\varphi_X^{(2)}$, and $\widehat\varphi_X^{(4)}$ can be computed based on \cite{li/vuong:1998} and its extensions -- see Appendix \ref{sec:li_vuong} and Appendix \ref{sec:derivative}, respectively.

{\bf Tolerance Level $\eta$:}
In light of Proposition \ref{prop:bounds}, we can choose the tolerance level $\eta$ in the following manner. 
By (\ref{eq:hermite_smoothness_moments}), we can satisfy condition (\ref{eq:approximation_psi}) of Theorem \ref{theorem:approximation} if $\frac{1.086435\pi^{-1/4}}{\sqrt{2q+3}}\sqrt{M\sum_{j=q+1}^{\infty}(2j+1)^{-3}} \leq \eta$.
Thus, having selected the smoothness bound $M$ and the sieve dimension $q$, one can set $\eta$ to
\begin{align}\label{eq:plug_in_eta}
\eta =
\frac{1.086435\pi^{-1/4}}{\sqrt{2q+3}}\sqrt{M\sum_{j=q+1}^{\infty}(2j+1)^{-3}}.
\end{align}


{\bf Tolerance Levels $\delta(t_1),\dots,\delta(t_L)$:}
In light of Proposition \ref{prop:bounds}, we can choose the tolerance levels $\delta(t_1),\dots,\delta(t_L)$ of approximation errors in the following manner. 
By (\ref{eq:hermite_smoothness_moments}), we can satisfy condition (\ref{eq:approximation_moments}) of Theorem \ref{theorem:approximation} if $\frac{1.086435\pi^{-1/4}}{\sqrt{2\pi}}\cdot\sqrt{M\sum_{j=q+1}^{\infty}(2j+1)^{-3}}\left(\frac{\mathbb{E}_P\left[|Y_1|\right]}{\sqrt{2q+3}}+1\right)\leq \delta(t_l)$ for all $l \in \{1,\dots,L\}$.
Thus, having selected the smoothness bound $M$ and the sieve dimension $q$, one can set $\delta(t_l)$ to
\begin{align}
\delta(t_l) =\frac{1}{\sqrt{2\pi}}\cdot\sqrt{M\sum_{j=q+1}^{\infty}(2j+1)^{-3}}\left(\frac{\mathbb{E}_n\left[|Y_1|\right]}{\sqrt{2q+3}}+1\right)
\end{align}
for all $l \in \{1,\dots,L\}$, where we replaced $1.086435\pi^{-1/4}$ by one for slackness accounting for estimation of $\mathbb{E}_P\left[|Y_1|\right]$ by $\mathbb{E}_n\left[|Y_1|\right]$.

{\bf Sieve Dimension $q$}: The sieve dimension $q$ can be chosen by adapting a bandwidth selection method suggested in \cite{bissantz/lutz/holzmann/munk:2007} in density deconvolution with known error distribution; similar bandwidth selection rules are also used in \cite{kato/sasaki:2018} and \cite{adusumilli/otsu/whang:2017} in the deconvolution literature. 
Given $q$, choose tolerance levels $\eta, \delta (t_{1}),\dots,\delta(t_{L})$ depending on $q$, and then construct the upper and lower functions $f^{U}(x) = f_{q}^{U}(x)$ and $f^{L}(x) = f_{q}^{L}(x)$ according to Algorithm \ref{algorithm:main}.
Then use the midpoint  $\hat{f}_{q}(x) = \{ f_{q}^{U}(x) + f_{q}^{L}(x) \}/2$  as a surrogate of a point estimate of $f(x)$.
Realizing that a sieve dimension corresponds to the reciprocal of a bandwidth, we suggest the following rule to choose $q$. Construct a candidate set for $q$ as $\{ q_{\min}, \dots, q_{\max} \}$, and compute the $L^{\infty}$-distance  between the density estimates with adjacent sieve dimensions, $d_{q,q+1}^{(\infty)} = \sup_{x \in I} |\hat{f}_{q+1}(x) - \hat{f}_{q}(x)|$. Then we choose the smallest $q$ such that $d_{q,q+1}^{(\infty)}$ is larger than $\rho d_{q_{\min},q_{\min + 1}}^{(\infty)}$ for some $\rho > 1$ (or alternatively we can choose the largest $q$ such that $d_{q,q+1}^{(\infty)}$ is smaller  than $\rho d_{q_{\min},q_{\min + 1}}^{(\infty)}$). In practice, it is  recommended
to make use of visual information on how $d_{q,q+1}^{(\infty)}$ behaves as $q$ decreases when
determining the sieve dimension.

\subsection{Implementation}\label{sec:implementation}

When we implement the Anderson-Rubin-type inference, we would generally sweep across the parameter set $\Theta^{q+1}$ of sieve coefficients, and this operation may demand long computational time.
However, we do not need to conduct the test at every point in $\Theta^{q+1}$, because properties of probability density functions and characteristic functions together with the sieve approximation property rule out substantially many elements of $\Theta^{q+1}$.
From the property $f \geq 0$ of probability density functions and the approximation bound (\ref{eq:hermite_smoothness_psi}), we can impose the restriction
\begin{equation}\label{eq:constraint_nonnegative_density}
\boldsymbol{\psi}(x)^T \boldsymbol{\theta} \geq -\frac{1.086435\pi^{-1/4}}{\sqrt{2q+3}}\sqrt{M\sum_{j=q+1}^{\infty}(2j+1)^{-3}}
\qquad\text{for all } x \in I.
\end{equation}
Similarly, from the property $\left[ \mathcal{F}f \right](0) =1$ of characteristic functions and the approximation bound (\ref{eq:hermite_smoothness_phi}), we can impose the restriction
\begin{equation}\label{eq:constraint_characteristic_function}
\left| \left[\mathcal{F} \boldsymbol{\psi}^T \boldsymbol{\theta}\right] (0) -1 \right| \leq \frac{1.086435\pi^{-1/4}\sqrt{2\pi}}{\sqrt{2q+3}}\sqrt{M\sum_{j=q+1}^{\infty}(2j+1)^{-3}}.
\end{equation}
Since the Hermite function is an eigenfunction of $\mathcal{F}$, (\ref	
{eq:constraint_characteristic_function}) can be simplified when the Hermite orthonormal basis $\boldsymbol{\Psi}$ (see Section \ref{sec:sieve}) is used.
Specifically, (\ref{eq:constraint_characteristic_function}) reduces to
\begin{equation}\label{eq:constraint_characteristic_function_hermite}
\left| \sqrt{2\pi} \boldsymbol{\psi}(0)^T \text{diag}\left(1,0,-1,0\dots\right) \boldsymbol{\theta} -1 \right| \leq \frac{1.086435\pi^{-1/4}\sqrt{2\pi}}{\sqrt{2q+3}}\sqrt{M\sum_{j=q+1}^{\infty}(2j+1)^{-3}}.
\end{equation}
Note that the left-hand side of (\ref{eq:constraint_characteristic_function_hermite}) does not require to compute an integral unlike that of (\ref{eq:constraint_characteristic_function}), which is a major advantage of using the Hermite orthonormal basis in the context of deconvolution.
Use of the constraints (\ref{eq:constraint_nonnegative_density}) and (\ref{eq:constraint_characteristic_function})/(\ref{eq:constraint_characteristic_function_hermite}) is motivated by the definition of the confidence band (\ref{eq:confidence_band}) consisting only of ``density functions'' $f \in \mathcal{L}$ which indexes the set $\mathbf{B}_{q+1,\eta}(f)$ of possible values of $\boldsymbol{\theta}$.

We also remark that we do not need to conduct a grid search for the purpose of drawing confidence bands.
In fact, solving $\min_{\boldsymbol{\theta} \in \Theta^{q+1}} \boldsymbol{\psi}(x)^T \boldsymbol{\theta}$ subject to $T(\boldsymbol{\theta}) \leq c(\alpha)$ as well as (\ref{eq:hermite_theta}), (\ref{eq:constraint_nonnegative_density}), and (\ref{eq:constraint_characteristic_function})/(\ref{eq:constraint_characteristic_function_hermite}) yields the lower bound of the confidence band up to the approximation error $\eta$.
Similarly, solving $\max_{\boldsymbol{\theta} \in \Theta^{q+1}} \boldsymbol{\psi}(x)^T \boldsymbol{\theta}$ subject to $T(\boldsymbol{\theta}) \leq c(\alpha)$ as well as (\ref{eq:hermite_theta}), (\ref{eq:constraint_nonnegative_density}), and (\ref{eq:constraint_characteristic_function})/(\ref{eq:constraint_characteristic_function_hermite}) yields the upper bound of the confidence band up to the approximation error $\eta$.
Accounting for these points, we propose the following implementation algorithm.

\begin{algorithm}\label{algorithm:appendix}${}$
\begin{enumerate}
\item Choose the tuning parameters according to Section \ref{sec:tuning_parameters}.
\item For each $x \in I$, compute $f^L(x) = \min_{\boldsymbol{\theta} \in \Theta^{q+1}} \boldsymbol{\psi}(x)^T \boldsymbol{\theta}$ subject to $T(\boldsymbol{\theta}) \leq c(\alpha,\boldsymbol{\theta})$, (\ref{eq:hermite_theta}), (\ref{eq:constraint_nonnegative_density}), \& (\ref{eq:constraint_characteristic_function_hermite}).
\item For each $x \in I$, compute $f^U(x) = \max_{\boldsymbol{\theta} \in \Theta^{q+1}} \boldsymbol{\psi}(x)^T \boldsymbol{\theta}$ subject to $T(\boldsymbol{\theta}) \leq c(\alpha,\boldsymbol{\theta})$, (\ref{eq:hermite_theta}), (\ref{eq:constraint_nonnegative_density}), \& (\ref{eq:constraint_characteristic_function_hermite}).
\item The confidence band is set to $\left[ f^L(x) - \eta, \ f^U(x) + \eta \right]$, $x \in I$.
\end{enumerate}
\end{algorithm}

We remark that, in computing the test statistic $T(\boldsymbol{\theta})$ in steps 2 and 3 of the algorithm above, the use of the Hermite orthonormal basis element $\psi = \psi_j$, $j=0,1,\dots$ simplifies (\ref{eq:def_R}) and (\ref{eq:def_I}) to
\begin{align*}
{R}_{\psi_j,t}(y_1,y_2)=\sqrt{2\pi} &\left\{-\cos(ty_2)(y_1\mathrm{Im}(i^j {\psi}_j(t))+\mathrm{Re}(i^j {\psi}_j^{(1)}(t)))\right.
\notag\\
& \left. \ \ -\sin(ty_2)(y_1\mathrm{Re}(i^j {\psi}_j(t))-\mathrm{Im}(i^j {\psi}_j^{(1)}(t))) \right\}
&&\text{and}
\\
{I}_{\psi_j,t}(y_1,y_2)= \sqrt{2\pi} &\left\{\cos(ty_2)(y_1\mathrm{Re}(i^j {\psi}_j(t))+\mathrm{Im}(i^j {\psi}_j^{(1)}(t)))\right.
\notag\\
& \left. \ \ -\sin(ty_2)(y_1\mathrm{Im}(i^j {\psi}_j(t))-\mathrm{Re}(i^j {\psi}_j^{(1)}(t))) \right\},
&&
\end{align*}
respectively.
As such, one need not compute an integral to obtain the test statistic $T(\boldsymbol{\theta})$.
This convenient property again follows from the fact that the Hermite function is an eigenfunction of $\mathcal{F}$.

\section{Simulation Studies}\label{sec:simulation}
In this section, we present and discuss finite-sample performance of the proposed method by simulation studies.
Simulation outcomes that we present include the size under the null of the true distribution, the power under alternative distributions, and the lengths of confidence bands.
The lengths will be further decomposed into the bias bound $\eta$ and the remaining lengths due to the stochastic part.

\subsection{Simulation Setting}

We employ three distribution families to generate the latent variable $X$ -- the normal distribution, the skew normal distribution, and the $t$ distribution.
We employ the skew normal distribution and the $t$ distribution to see whether our method is effective for asymmetric distributions and super-Gaussian tails, respectively.
Specifically, we generate a random sample of $(X,U_1,U_2)$ mutually independently according to the marginal laws:
\begin{align*}
\text{Model 1:} \quad & X \sim N(\xi_1,\xi_2^2), & U_1 \sim N(0,\sigma_{U_1}^2), & \quad U_2 \sim N(0,\sigma_{U_2}^2) &
\\
\text{Model 2:} \quad & X \sim SN(\xi_1,\xi_2,\xi_3), & U_1 \sim N(0,\sigma_{U_1}^2), & \quad U_2 \sim N(0,\sigma_{U_2}^2) &
\\
\text{Model 3:} \quad & X \sim t_{\xi_4}, & U_1 \sim N(0,\sigma_{U_1}^2), & \quad U_2 \sim N(0,\sigma_{U_2}^2) &
\end{align*}
Here, $N(\xi_1,\xi_2^2)$ denotes the normal distribution with mean $\xi_1$ and variance $\xi_2^2$, $SN(\xi_1,\xi_2,\xi_3)$ denotes the skew normal distribution with location $\xi_1$, scale $\xi_2$, and shape $\xi_3$, and $t_{\xi_4}$ denotes the $t$ distribution with $\xi_4$ degrees of freedom.
The distribution parameters for the latent variable $X$ are set to $(\xi_1,\xi_2)=(0,1)$ for Model 1, $(\xi_1,\xi_2,\xi_3)=(0,1,1)$ for Model 2, and $\xi_4=5$ for Model 3.
The choice of the normal error distribution, which is an instance of super-smooth distributions, imposes a difficult case in deconvolution -- see \cite{li/vuong:1998}.
The error variance parameters are set to $\sigma_{U_1} = \sigma_{U_2} = 0.5$ in each of the three models.
We conduct experiments with three sample sizes $n=$ 250, 500, and 1,000, and run 2,500 Monte Carlo iterations for each set of simulations.

We follow the practical guideline provided in Appendix \ref{sec:practical} to construct confidence bands.
As remarked previously, the smoothness bound $M \in (0,\infty)$ as well as the sieve dimension $q \in \mathbb{N}$ can be imposed by a researcher in the spirit of honest inference \citep[cf.][]{armstrong/kolesar:2018} -- see Section \ref{sec:tuning_parameters}.
We experiment with the tuning parameters $q \in \{5,7,9\}$.
The function classes are defined by (\ref{eq:hermite_smoothness}) with $M=15$, $25$, and $35$ for Models 1, 2, and 3, respectively.
The frequency bound is set to $T=5$ and the number of frequency grid points is set to $L=50$.
The interval on which the confidence band is formed is set to $I = \left[\mathbb{E}[X]-2\sqrt{\mathrm{Var}(X)},\mathbb{E}[X]+2\sqrt{\mathrm{Var}(X)}\right]$, where $\mathbb{E}[X]$ and $\mathrm{Var}(X)$ are the theoretical mean and the theoretical variance, respectively, of $X$ under the relevant model.
To enjoy favorable speed of computation for numerous Monte Carlo iterations, we use the conservative critical value $c(\alpha)$. The level is set to $\alpha = 0.05$ throughout.

\subsection{Simulation Results}

Figure \ref{fig:model1} (A) shows the simulated frequencies that the confidence band formed under Model 1 covers alternative probability density functions for $N(\xi_1,\xi_2^2)$ indexed by location parameter values $\xi_1 \in [0.0,1.0]$ while the scale parameter is fixed at the true value $\xi_2=1.0$.
The coverage frequency under $\xi_1=0.0$ indicate (the complement of) the size, whereas the coverage frequencies under $\xi_1 \in (0.0,1.0]$ indicate (the complement of) the power.
Similarly, Figure \ref{fig:model1} (B) shows the simulated frequencies that the confidence band formed under Model 1 covers alternative probability density functions for $N(\xi_1,\xi_2^2)$ indexed by scale parameter values $\xi_2 \in [1.0,2.0]$ while the location parameter is fixed at the true value $\xi_1=0.0$.
These results show the correct size and increasing power.
The size entails over-coverage, which is still consistent with our theory on size control.

\begin{figure}[tbp]
	\centering
		(A) Coverage frequencies with $q=7$ under location alternatives in Model 1\\
		\includegraphics[width=0.6\textwidth]{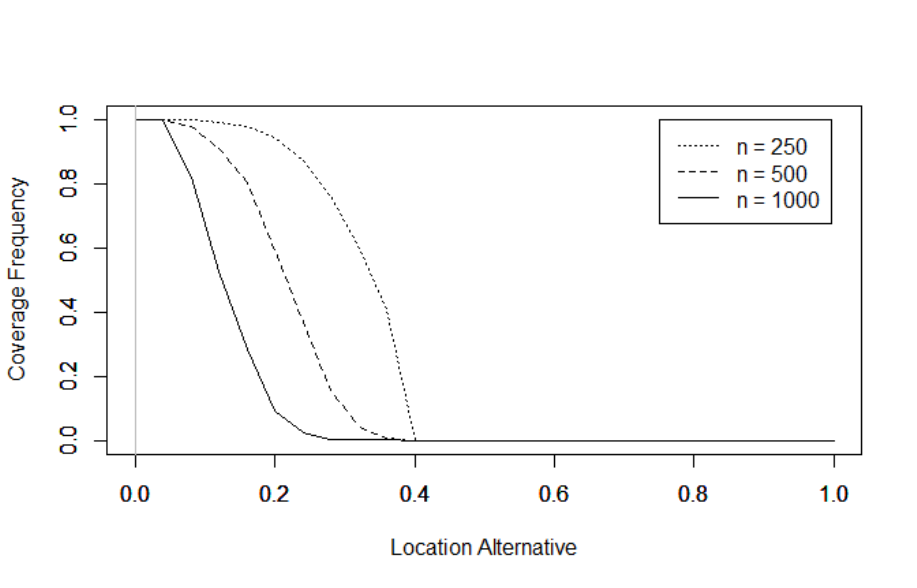}\bigskip\\
		(B) Coverage frequencies with $q=7$ under scale alternatives in Model 1\\
		\includegraphics[width=0.6\textwidth]{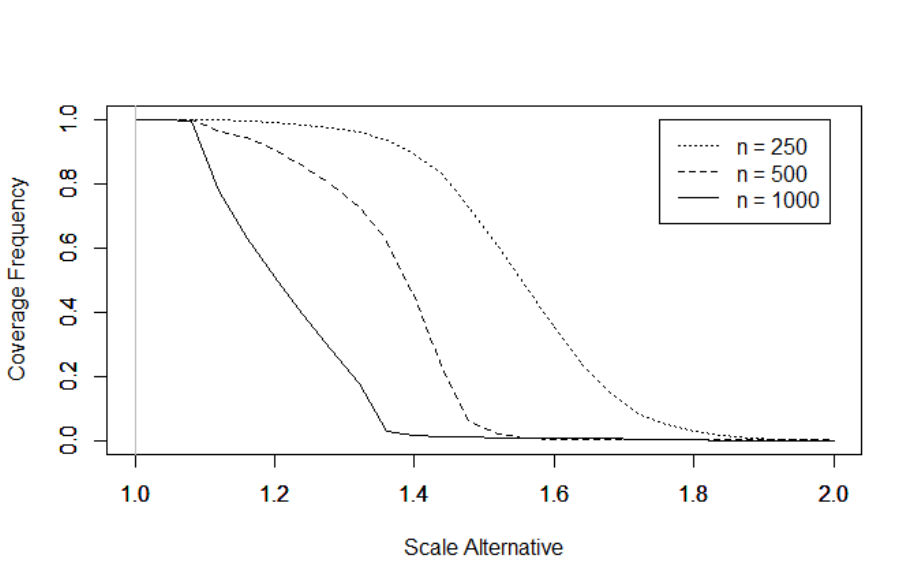}
	\caption{The simulated frequencies that the confidence band formed under Model 1 covers alternative probability density functions for $N(\xi_1,\xi_2^2)$ with the tuning parameter $q=7$. Panel (A) runs across alternative location parameter values $\xi_1 \in [0.0,1.0]$ while the scale parameter is fixed at the true value $\xi_2=1.0$. Panel (B) runs across alternative scale parameter values $\xi_2 \in [1.0,2.0]$ while the location parameter is fixed at the true value $\xi_1=0.0$.}
	\label{fig:model1}
\end{figure}

\begin{figure}[tbp]
	\centering
		(A) Coverage frequencies with $q=7$ under location alternatives in Model 2\\
		\includegraphics[width=0.6\textwidth]{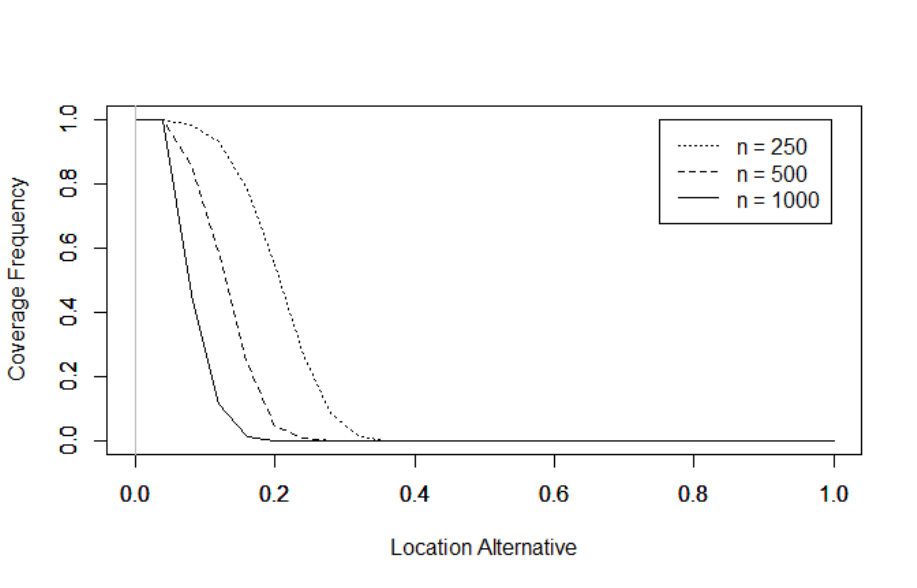}\bigskip\\
		(B) Coverage frequencies with $q=7$ under location alternatives in Model 3\\
		\includegraphics[width=0.6\textwidth]{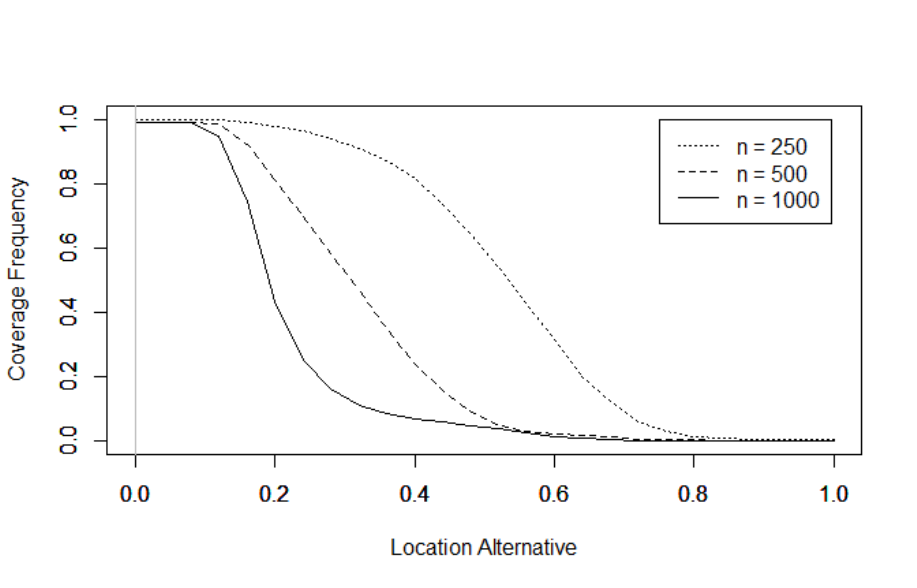}
		\caption{(A) The simulated frequencies that the confidence band formed under Model 2 covers alternative probability density functions for $SN(\xi_1,\xi_2,\xi_3)$ indexed by the alternative location parameter values $\xi_1 \in [0.0,1.0]$ while the scale and shape parameters are fixed at $(\xi_2,\xi_3)=(1.0, 1.0)$. (B) The simulated frequencies that the confidence band formed under Model 3 covers alternative probability density functions for the (non-) central t-distributions indexed by the alternative locations in $[0.0,1.0]$ while the degrees of freedom is fixed at the true value $5$. Results in both panel (A) and panel (B) are based on the tuning parameter $q=7$}
	\label{fig:model2model3}
\end{figure}

Figures \ref{fig:model2model3} (A) and \ref{fig:model2model3} (B) show analogous results to Figure \ref{fig:model1} (A) except that Model 2 and Model 3, respectively, are used instead of Model 1.
For Model 2, the shape parameter is fixed at the true value $\xi_3=1.0$.
These results evidence that the proposed method is similarly effective for the cases where the latent variable $X$ follows asymmetric distributions or distributions with super-Gaussian tails.



We next present average lengths of the confidence bands on $I$, and their decomposition into the bias bound $\eta$ and the remaining length due to the stochastic part.
Table \ref{tab:lengths} summarizes results on the lengths.
There are a couple of features in these results that deserve discussions.
First, observe that the confidence bands shrink as sample size increases when the sieve dimension $q$ is held fixed.
On the other hand, the bias bound $\eta$ remains invariant across sample sizes while the sieve dimension $q$ is fixed.
These results are natural features of our approach.
Second, observe that the bias bound $\eta$ decreases as the sieve dimension $q$ increments.

\begin{table}[tbp]
	\centering
	\scalebox{1}{
		\begin{tabular}{cccccccc}
		\hline\hline
		  &&     &     && Average & Supremum & Stochastic \\
			&& $q$ & $n$ && Length & Bias ($2\eta$) & Length \\
		\hline
			Model 1 && 5 &   250 && 0.290 & 0.073 & 0.217\\
			        && 5 &   500 && 0.237 & 0.073	& 0.164\\
							&& 5 & 1,000 && 0.193 & 0.073	& 0.120\\
		\cline{3-8}
			Model 1 && 7 &   250 && 0.263 & 0.048	& 0.215\\
			        && 7 &   500 && 0.196 & 0.048 & 0.148\\
							&& 7 & 1,000 && 0.144 & 0.048	& 0.097\\
		\cline{3-8}
			Model 1 && 9 &   250 && 0.235 & 0.034 & 0.200\\
			        && 9 &   500 && 0.168 & 0.034 & 0.134\\
							&& 9 & 1,000 && 0.120 & 0.034 & 0.086\\
		\hline
			Model 2 && 5 &   250 && 0.311 & 0.094	& 0.217\\
			        && 5 &   500 && 0.244 & 0.094	& 0.150\\
							&& 5 & 1,000 && 0.194 & 0.094 & 0.100\\
		\cline{3-8}
			Model 2 && 7 &   250 && 0.288 & 0.062	& 0.226\\
			        && 7 &   500 && 0.214 & 0.062	& 0.152\\
							&& 7 & 1,000 && 0.155 & 0.062	& 0.094\\
		\cline{3-8}
			Model 2 && 9 &   250 && 0.255 & 0.044 & 0.210\\
			        && 9 &   500 && 0.186 & 0.044 & 0.142\\
							&& 9 & 1,000 && 0.126 & 0.044 & 0.081\\
		\hline
			Model 3 && 5 &   250 && 0.349 & 0.111 & 0.238\\
			        && 5 &   500 && 0.291 & 0.111 & 0.180\\
							&& 5 & 1,000 && 0.252 & 0.111 & 0.140\\
		\cline{3-8}
			Model 3 && 7 &   250 && 0.341 & 0.073	& 0.268\\
			        && 7 &   500 && 0.260 & 0.073	& 0.187\\
							&& 7 & 1,000 && 0.205 & 0.073	& 0.132\\
		\cline{3-8}
			Model 3 && 9 &   250 && 0.351 & 0.053 & 0.299\\
			        && 9 &   500 && 0.242 & 0.053 & 0.190\\
							&& 9 & 1,000 && 0.179 & 0.053 & 0.126\\
		\hline\hline
		\end{tabular}
	}
	\caption{Average lengths of confidence bands, the supremum biases ($\eta$), and the lengths for stochastic parts of confidence bands.}
	\label{tab:lengths}
\end{table}

\begin{figure}[tbp]
	\centering
	\begin{tabular}{ccc}
		& $n=$ 250\\
		Model 1 & Model 2 & Model 3\\
		\includegraphics[width=0.3\textwidth]{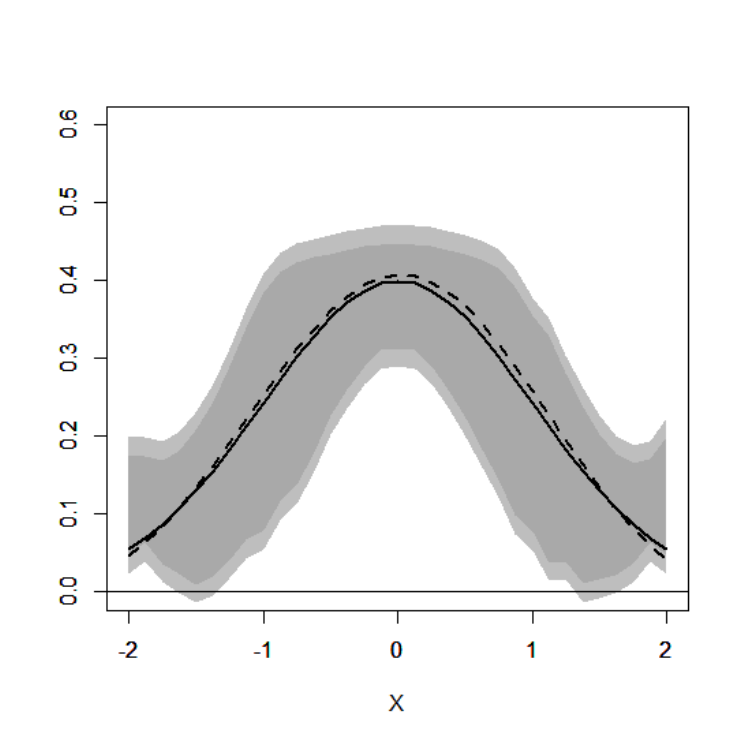}&
		\includegraphics[width=0.3\textwidth]{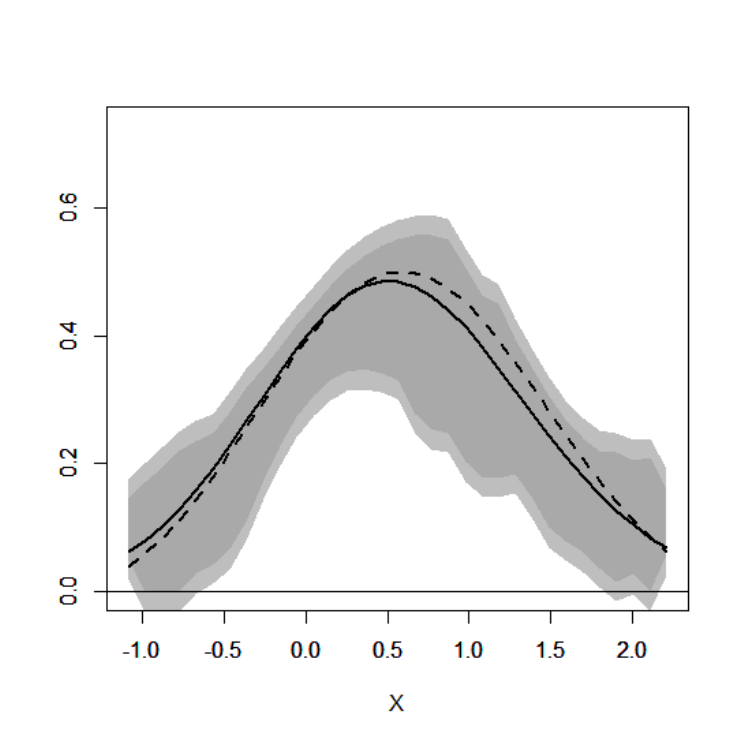}&
		\includegraphics[width=0.3\textwidth]{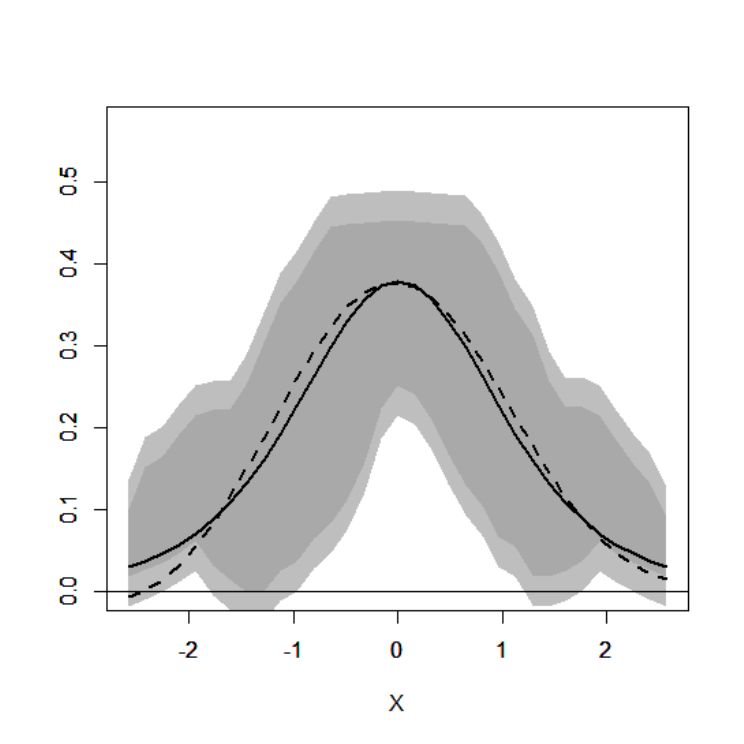}\\
		& $n=$ 500\\
		Model 1 & Model 2 & Model 3\\
		\includegraphics[width=0.3\textwidth]{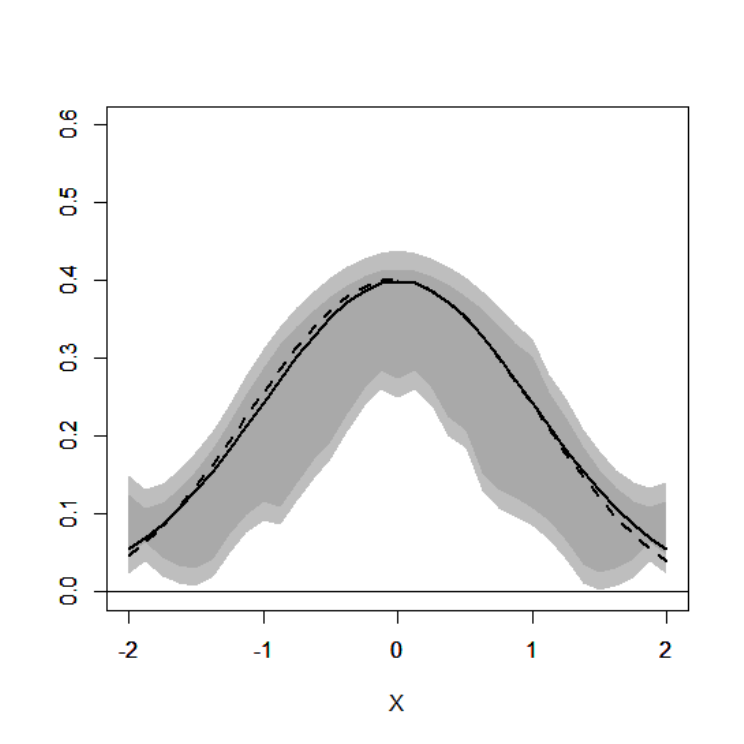}&
		\includegraphics[width=0.3\textwidth]{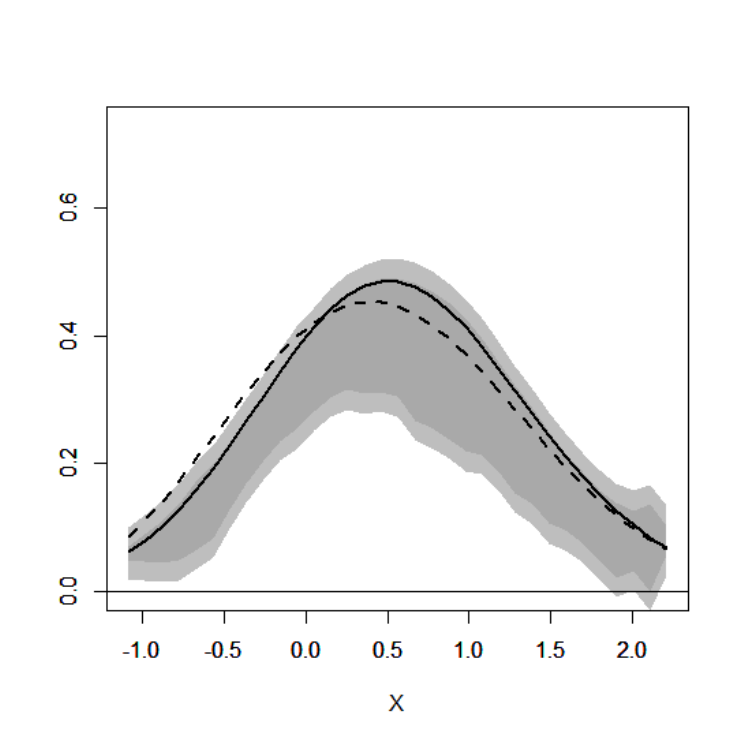}&
		\includegraphics[width=0.3\textwidth]{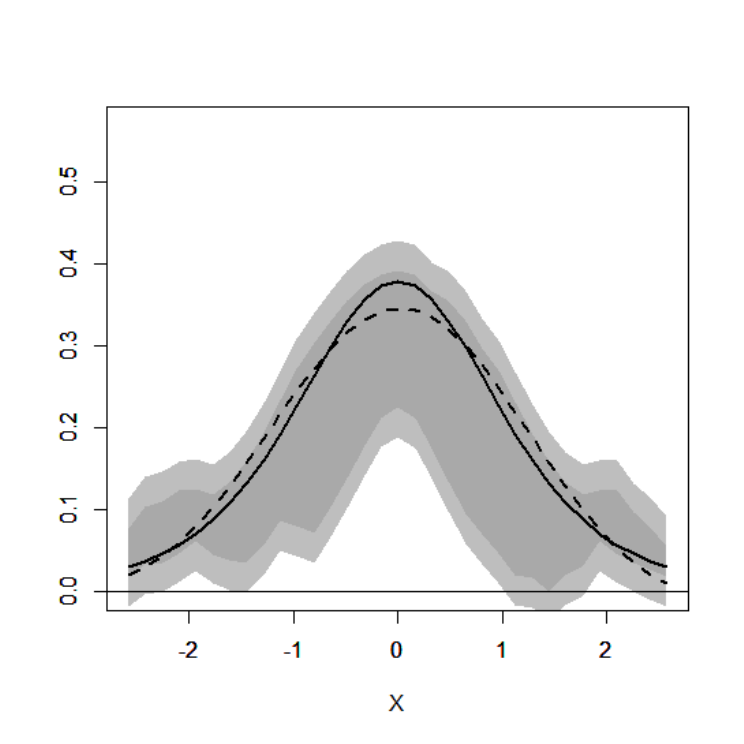}\\
		& $n=$ 1,000\\
		Model 1 & Model 2 & Model 3\\
		\includegraphics[width=0.3\textwidth]{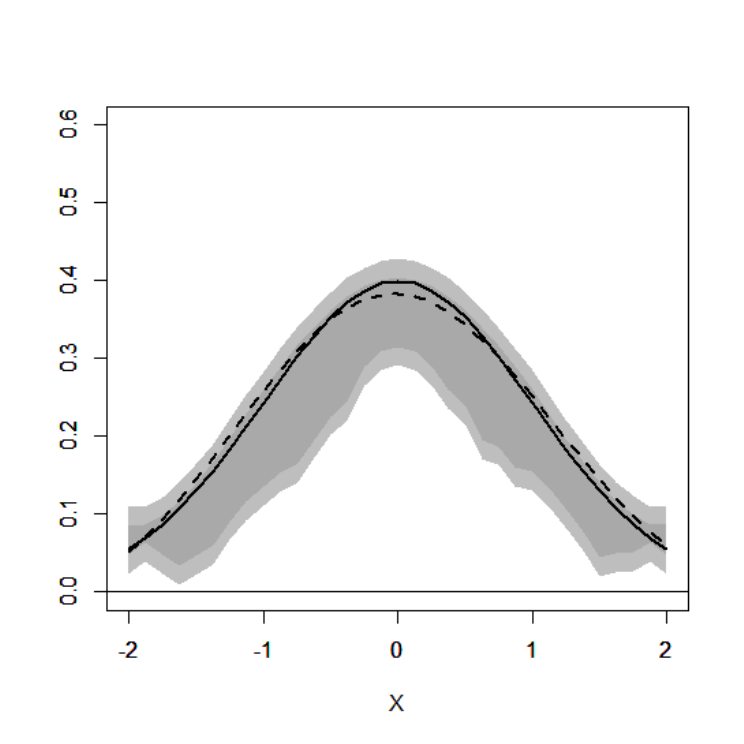}&
		\includegraphics[width=0.3\textwidth]{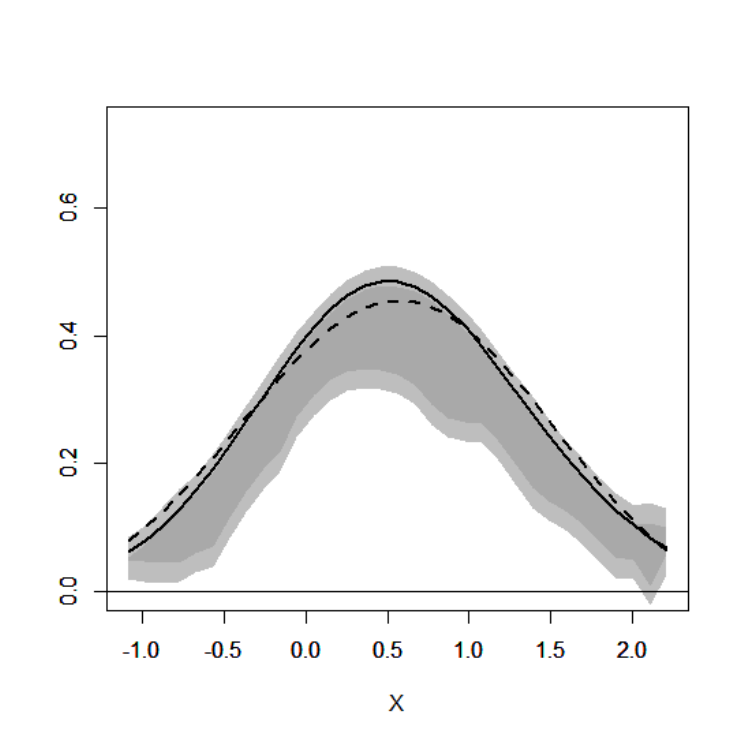}&
		\includegraphics[width=0.3\textwidth]{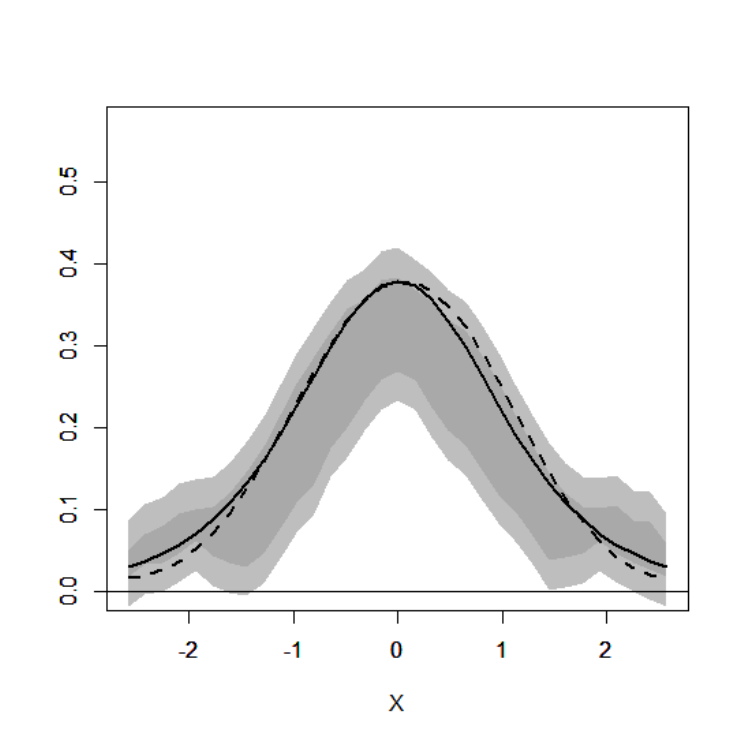}
	\end{tabular}
	\caption{Instances of confidence bands. The gray shades indicate the confidence bands. The internal dark gray shades indicate the stochastic parts of the confidence bands. The solid and dashed curves indicate the true density functions and Li-Vuong estimates, respectively.}
	\label{fig:instances}
\end{figure}

Finally, we display instances of confidence bands in Figure \ref{fig:instances}.
The gray shades indicate the confidence bands including the bias bound and the stochastic parts together.
The internal dark gray shades include only the stochastic parts.
We also plot the true density functions and Li-Vuong estimates (with the choice of tuning parameter according to Appendix \ref{sec:tuning_parameters_li_vuong}) as solid and dashed curves, respectively.
While such instances of confidence bands will not tell us any evidence on the statistical properties, they at least inform how a confidence band may look in applications.

\subsection{Simulations with Nonparametric Bootstrap}
In this paper, we propose a novel inference method based on Kotlarski's identity.
Under the absence of an inference method prior to our present paper, some existing applied papers that use Kotlarski's identity conduct the nonparametric bootstrap of the Li-Vuong estimator to draw confidence intervals though there is no theoretical guarantee for such a method to work.
In the current subsection, we use simulation studies to assess the coverage performance for such a bootstrap approach.
For the purpose of comparison, we use the same data generating processes as in the previous subsection, namely Model 1, Model 2, and Model 3.

Table \ref{tab:naive_bootstrap} summarizes the simulated frequencies that the na\"ive bootstrap confidence interval covers the true probability density function at locations $x \in \mathbb{E}[X]+\{-2,1,0,1,2\}\sqrt{\mathrm{Var}(X)}$.
The tuning parameter is selected based on a commonly used approach in this literature -- see Appendix \ref{sec:tuning_parameters_li_vuong}.
The coverage frequencies are close to the nominal probabilities only near the tails of the distributions, e.g., $x = \mathbb{E}[X] \pm 2\sqrt{\mathrm{Var}(X)}$, under Model 1 and Model 2.
On the other hand, the na\"ive bootstrap suffers from under-coverage near the center of the distributions under Model 1 and Model 2.
Furthermore, the na\"ive bootstrap suffers from even more severe under-coverage both near the center of the distribution and near the tails of the distribution under Model 3.
These results suggest that one should substitute our proposed method for the traditional na\"ive bootstrap approach.

\begin{table}[tbp]
	\centering
		\begin{tabular}{lccccccc}
		\hline\hline
		          &       && \multicolumn{5}{c}{$(x-\mathbb{E}[X])/\sqrt{\mathrm{Var}(X)}$}\\
		\cline{4-8}
		          & $n$   && $-2$ & $-1$ & $0$ & $1$ & $2$\\
		\hline
			Model 1 & 1,000 && 0.951 & 0.911 & 0.816 & 0.904 & 0.946\\
			        & 2,000 && 0.952 & 0.887 & 0.792 & 0.884 & 0.957\\
			        & 4,000 && 0.947 & 0.861 & 0.743 & 0.861 & 0.955\\
		\hline
			Model 2 & 1,000 && 0.949 & 0.929 & 0.848 & 0.885 & 0.956\\
			        & 2,000 && 0.951 & 0.908 & 0.796 & 0.862 & 0.940\\
			        & 4,000 && 0.946 & 0.888 & 0.748 & 0.820 & 0.934\\
		\hline
			Model 3 & 1,000 && 0.576 & 0.519 & 0.397 & 0.559 & 0.602\\
			        & 2,000 && 0.477 & 0.467 & 0.287 & 0.456 & 0.465\\
			        & 4,000 && 0.416 & 0.383 & 0.172 & 0.388 & 0.416\\
		\hline\hline
		\end{tabular}
	\caption{The simulation frequencies that the traditional na\"ive bootstrap confidence interval covers the true probability density function at locations $x \in \mathbb{E}[X]+\{-2,1,0,1,2\}\sqrt{\mathrm{Var}(X)}$.}
	\label{tab:naive_bootstrap}
\end{table}

\section{Conclusion}\label{sec:conclusion}
Since its introduction to econometrics by \citet*{li/vuong:1998}, Kotlarski's identity \citep*{kotlarski1967} -- see also \citet{rao:1992} -- has been widely used in empirical economics.
Examples include applications to empirical auctions \citep[e.g.,][]{LiPerrigneVuong2000,krasnokutskaya:2011}, income dynamics \citep[e.g.,][]{bonhomme/robin:2010}, and labor economics \citep[e.g.,][]{cunha/heckman/navarro:2005,cunha/heckman/schennach:2010,bonhomme/sauder:2011,kennan/walker:2011}.
Despite its popular use in applications, a method of inference based on Kotlarski's identity has long been missing in the literature.
After twenty years since \citet*{li/vuong:1998}, we now propose a method of inference based on Kotlarski's identity.
Specifically, we develop confidence bands for the probability density function $f_X$ of $X$ in the repeated measurement model where two measurements $(Y_1,Y_2)$ of unobserved variable $X$ are available in data with additive independent errors, $U_1 = Y_1-X$ and $U_2 = Y_2-X$.

Our construction of confidence bands can be summarized as follows.
First, we derive linear complex-valued moment restrictions based on Kotlarski's identity.
Second, we let the Hermite polynomial sieve approximate unknown probability density functions.
Third, for a given sieve dimension and for a given class for probability density functions, we compute a bias bound for the linear complex-valued moment restrictions, and slack the linear complex-valued moment restrictions by this bias bound.
Fourth, we compute the uniform norm of the self-normalized process of the slacked linear complex-valued moment restrictions as the test statistics for each point in a set of sieve coefficients.
Fifth, inverting this test statistic yields a confidence set of sieve approximations to possible probability density functions.
Sixth, for a given sieve dimension and for a given class for probability density functions, we compute a bias bound for sieve approximations of probability density functions, and the desired confidence band is obtained by uniformly enlarging the set of sieve approximations by this bias bound.

We not only provide a method that works, but also care for its practicality.
The Fourier transform and the inverse Fourier transform operations are known to be computationally costly in the deconvolution literature.
By exploiting the property of the Hermite functions as eigen-functions of the Fourier transform operator, we propose to let the Hermite polynomial sieve approximate both the density and characteristic functions without having to implement numerical integrations within each iteration of a numerical optimization routine.
This convenient feature of the proposed method saves computational resources.
Furthermore, we also exploit a couple of other convenient properties of the Hermite functions (namely the Schr\"odinger equation for a harmonic oscillator and a pair of recursive equations), and consequently obtain informative bias bounds and thus informative inference.
With these practical features of our method, simulation studies indeed conclude reasonably fast with informative inference results.
The results evidence the efficacy of the proposed method.
Since Kotlarski's identity is one of the most popular methods in a number of applied fields, including empirical auctions, income dynamics, and labor economics, we hope that our method will contribute to the practice of economic analyses in these and other topics.


{\small
\bibliography{mybib}
}

\newpage
\appendix
\section*{Online Appendix}
\section{Proofs for the Main Theorems}

\subsection{Proof of Theorem \ref{theorem:moment_condition} (Linear Complex-Valued Moment Restrictions)}\label{sec:theorem:moment_condition}
\begin{proof}
By (\ref{eq:main}) and Assumption \ref{a:continuous_independence} (ii), we have
\begin{equation*}
\mathbb{E}_P\left[\exp(it_1Y_1+it_2Y_2)\right]
=
\phi_X(t_1+t_2)\phi_{U_1}(t_1)\phi_{U_2}(t_2)
\end{equation*}
for every $(t_1,t_2) \in \mathbb{R}^2$.
Note that random variables with finite first moments have continuously differentiable characteristic functions.
Since $\phi_{U_1}(0)=1$ and $\phi_{U_1}^{(1)}(0)=0$ by Assumption \ref{a:continuous_independence} (i), it follows that 
\begin{align*}
\mathbb{E}_P\left[\exp(itY_2)\right]
&=
\phi_X(t)\phi_{U_2}(t)
\qquad\text{and}\\
\left. \frac{\partial}{\partial t_1}\mathbb{E}_P\left[\exp(it_1Y_1+itY_2)\right] \right\vert_{t_1=0}
&=
\phi_X^{(1)}(t)\phi_{U_2}(t)
\end{align*}
for every real $t$ under Assumption \ref{a:continuous_independence} (i).
Therefore, 
\begin{align*}
\mathbb{E}_P\left[\left(iY_1\phi_X(t)-\phi_X^{(1)}(t)\right)\exp(itY_2)\right]
&=
\mathbb{E}_P\left[iY_1\exp(itY_2)\right]\phi_X(t)-\mathbb{E}_P\left[\exp(itY_2)\right]\phi_X^{(1)}(t)\\
&=
\frac{\partial}{\partial t_1}\mathbb{E}_P\left[\exp(it_1Y_1+itY_2)\right]\mid_{t_1=0}\phi_X(t)-\mathbb{E}_P\left[\exp(itY_2)\right]\phi_X^{(1)}(t)\\
&=
\phi_X^{(1)}(t)\phi_{U_2}(t)\phi_X(t)-\phi_X(t)\phi_{U_2}(t)\phi_X^{(1)}(t)
=
0
\end{align*}
for every real $t$, and the claim of the lemma follows.
\end{proof}

\subsection{Proof of Theorem \ref{theorem:size_control} (Size Control)}\label{sec:theorem:size_control}
\begin{proof}[Proof of Theorem \ref{theorem:size_control}]
Let $P \in \mathcal{P}$ and $f \in \mathcal{L}_0^\ast(P)$.
We have 
\begin{eqnarray*}
\mathbb{P}_P(f\in \mathcal{C}_{n}(\alpha))
&=&
\mathbb{P}_P(T(\boldsymbol{\theta})\leq c(\alpha,\boldsymbol{\theta})\mbox{ for some }\boldsymbol{\theta}\in\mathbf{B}_{q+1,\eta}(f))\\
&\geq&
\mathbb{P}_P(T(\boldsymbol{\theta}_{\ast\ast})\leq c(\alpha,\boldsymbol{\theta}_{\ast\ast}))\\
\end{eqnarray*}
where $\boldsymbol{\theta}_{\ast\ast} \in \mathbf{B}_{q+1,\eta}(f)$ satisfies
$$
\max_{1\leq l\leq L}\left(\max\{|\mathbb{E}_P[\mathbf{R}_{t_l}]^T\boldsymbol{\theta}_{\ast\ast}|,|\mathbb{E}_P[\mathbf{I}_{t_l}]^T\boldsymbol{\theta}_{\ast\ast}|\}-\delta(t_l)\right) \leq 0
$$
as $f \in \mathcal{L}_0^{\ast}(P)$.
By \citet*[][Theorem A.1]{chernozhukov/chetverikov/kato:2017}, there exist positive constants $c$ and $C$ depending on $c_1$ and $C_1$ under Assumption \ref{a:size_control} (it is not difficult to see that $c$ and $C$ are independent of $\alpha$; see Theorem 4.3 in \cite{chernozhukov/chetverikov/kato:2017}) such that
$$
\mathbb{P}_P(T(\boldsymbol{\theta}_{\ast\ast}) \leq c(\alpha,\boldsymbol{\theta}_{\ast\ast}))\geq 1-\alpha-Cn^{-c}.
$$
Therefore, the statement of the theorem follows. 
\end{proof}

\subsection{Proof of Theorem \ref{theorem:approximation} (Approximation)}\label{sec:theorem:approximation}
\begin{proof}
Let $f\in\mathcal{L}_0(P)$ and $\phi = \mathcal{F} f$.
Assumption \ref{a:orthonormal} implies that
\begin{align*}
f = \sum_{j=0}^{\infty} \langle{ f,\psi_j }\rangle \cdot \psi_j
\qquad\text{and}\qquad
\phi = \sum_{j=0}^{\infty} \langle{ f,\psi_j }\rangle \cdot \phi_j
\end{align*}
-- see \citet[][Theorem 5.27]{folland2007}.
Define 
\begin{align*}
\psi_{0:q}=\sum_{j=0}^q \langle{ f,\psi_j }\rangle \cdot \psi_j 
\qquad\text{and}\qquad
\phi_{0:q}=\sum_{j=0}^q \langle{ f,\psi_j }\rangle \cdot \phi_j.
\end{align*}
Since $f\in\mathcal{L}_0(P)$ and $\phi = \mathcal{F} f$,
\begin{eqnarray*}
\left|\mathbb{E}_P\left[\left(iY_1\phi_{0:q}(t)-\phi_{0:q}^{(1)}(t)\right)\exp(itY_2)\right]\right|
&=& 
\left|\mathbb{E}_P\left[\left(iY_1(\phi_{0:q}(t)-\phi(t))-(\phi_{0:q}^{(1)}(t)-\phi^{(1)}(t))\right)\exp(itY_2)\right]\right|\\
&=& 
\left|\sum_{j=q+1}^{\infty} \langle{ f,\psi_j }\rangle \cdot \mathbb{E}_P\left[\left(iY_1\phi_j(t)-\phi_j^{(1)}(t)\right)\cdot\exp(itY_2)\right]\right|.
\end{eqnarray*}
Similarly,
\begin{eqnarray*}
\sup_{t\in I}|f(t)-\phi_{0:q}(t)|
&\leq&
\sup_{t\in I}\left|\sum_{j=q+1}^{\infty} \langle{ f,\psi_j }\rangle \cdot \psi_j(t)\right|.
\end{eqnarray*}
Therefore, the statement of the theorem follows.
\end{proof}

\subsection{Proof of Theorem \ref{thm_power} (Power)}\label{sec:thm_power}
\begin{proof}
This proof focuses on the case in (\ref{local_alternatves1}). 
The proofs for the cases of (\ref{local_alternatves2})--(\ref{local_alternatves4}) are similar. 
By the definition of $\mathcal{C}_{n}(\alpha)$, we can write
\begin{eqnarray*}
\mathbb{P}_P(f\notin \mathcal{C}_{n}(\alpha))
&=& 
\mathbb{P}_P\left(T(\boldsymbol{\theta})>c(\alpha,\boldsymbol{\theta})\mbox{ for every }\boldsymbol{\theta}\in\mathbf{B}_{q+1,\eta}(f)\right)\\
&\geq& 
\mathbb{P}_P\left(T(\boldsymbol{\theta})>\sqrt{2\log(4L)}+\sqrt{2\log(1/\alpha)}\mbox{ for every }\boldsymbol{\theta}\in\mathbf{B}_{q+1,\eta}(f)\right)\\
&=& 
\mathbb{P}_P\left(\inf_{\boldsymbol{\theta}\in\mathbf{B}_{q+1,\eta}(f)}T(\boldsymbol{\theta})>\sqrt{2\log(4L)}+\sqrt{2\log(1/\alpha)}\right),
\end{eqnarray*}
where the inequality follows from 
$$
c(\alpha,\boldsymbol{\theta})\leq \sqrt{2\log(4L)}+\sqrt{2\log(1/\alpha)}
$$
-- see \citet*[][Lemma D.4]{chernozhukov/chetverikov/kato:2017}.
If $B_{V}\leq \nu$, then  
\begin{eqnarray*}
\inf_{\boldsymbol{\theta}\in\mathbf{B}_{q+1,\eta}(f)}T(\boldsymbol{\theta})
&\geq&
\inf_{\boldsymbol{\theta}\in\mathbf{B}_{q+1,\eta}(f)}\sqrt{n}\frac{\mathbb{E}_n[\mathbf{R}_{t_\ast}]^T\boldsymbol{\theta}-\delta(t_\ast)}{\sqrt{\boldsymbol{\theta}^T\mathbb{V}_n(\mathbf{R}_{t_\ast})\boldsymbol{\theta}}}\\
&\geq&
\inf_{\boldsymbol{\theta}\in\mathbf{B}_{q+1,\eta}(f)}\frac{\mathbb{G}_n[\mathbf{R}_{t_\ast}]^T\boldsymbol{\theta}}{\sqrt{\boldsymbol{\theta}^T\mathbb{V}_n(\mathbf{R}_{t_\ast})\boldsymbol{\theta}}}
+
\inf_{\boldsymbol{\theta}\in\mathbf{B}_{q+1,\eta}(f)}\sqrt{n}\frac{\mathbb{E}_P[\mathbf{R}_{t_\ast}]^T\boldsymbol{\theta}-\delta(t_\ast)}{\sqrt{\boldsymbol{\theta}^T\mathbb{V}_n(\mathbf{R}_{t_\ast})\boldsymbol{\theta}}}\\
&\geq&
-\sup_{\boldsymbol{\theta}\in\mathbf{B}_{q+1,\eta}(f)}\frac{|\mathbb{G}_n[\mathbf{R}_{t_\ast}]^T\boldsymbol{\theta}|}{\sqrt{\boldsymbol{\theta}^T\mathbb{V}_n(\mathbf{R}_{t_\ast})\boldsymbol{\theta}}}
+
\inf_{\boldsymbol{\theta}\in\mathbf{B}_{q+1,\eta}(f)}\sqrt{n}\frac{\mathbb{E}_P[\mathbf{R}_{t_\ast}]^T\boldsymbol{\theta}-\delta(t_\ast)}{\sqrt{\boldsymbol{\theta}^T\mathbb{V}_P(\mathbf{R}_{t_\ast})\boldsymbol{\theta}+\nu}}\\
&\geq&
-\sup_{\boldsymbol{\theta}\in\mathbf{B}_{q+1,\eta}(f)}\frac{|\mathbb{G}_n[\mathbf{R}_{t_\ast}]^T\boldsymbol{\theta}|}{\sqrt{\boldsymbol{\theta}^T\mathbb{V}_n(\mathbf{R}_{t_\ast})\boldsymbol{\theta}}}+\sqrt{2\log(4L)}+\sqrt{2\log(1/\alpha)}\\&&+(1+b) \cdot \mathbb{E}_P\left[\sup_{\boldsymbol{\theta}\in\mathbf{B}_{q+1,\eta}(f)}\frac{|\mathbb{G}_n[\mathbf{R}_{t_\ast}]^T\boldsymbol{\theta}|}{\sqrt{\boldsymbol{\theta}^T\mathbb{V}_n(\mathbf{R}_{t_\ast})\boldsymbol{\theta}}}\right]
\end{eqnarray*}
by (\ref{local_alternatves1}).
Thus, we obtain
\begin{eqnarray*}
\mathbb{P}_P(f\notin \mathcal{C}_{n}(\alpha))
&\geq& 
\mathbb{P}_P\left(\left\{\dfrac{\sup_{\boldsymbol{\theta}\in\mathbf{B}_{q+1,\eta}(f)}\frac{|\mathbb{G}_n[\mathbf{R}_{t_\ast}]^T\boldsymbol{\theta}|}{\sqrt{\boldsymbol{\theta}^T\mathbb{V}_n(\mathbf{R}_{t_\ast})\boldsymbol{\theta}}}}{\mathbb{E}_P\left[\sup_{\boldsymbol{\theta}\in\mathbf{B}_{q+1,\eta}(f)}\frac{|\mathbb{G}_n[\mathbf{R}_{t_\ast}]^T\boldsymbol{\theta}|}{\sqrt{\boldsymbol{\theta}^T\mathbb{V}_n(\mathbf{R}_{t_\ast})\boldsymbol{\theta}}}\right]}
<1+b\right\}\cap \{B_{V}\leq \nu\}\right)\\
&\geq& 
\mathbb{P}_P\left(\dfrac{\sup_{\boldsymbol{\theta}\in\mathbf{B}_{q+1,\eta}(f)}\frac{|\mathbb{G}_n[\mathbf{R}_{t_\ast}]^T\boldsymbol{\theta}|}{\sqrt{\boldsymbol{\theta}^T\mathbb{V}_n(\mathbf{R}_{t_\ast})\boldsymbol{\theta}}}}{\mathbb{E}_P\left[\sup_{\boldsymbol{\theta}\in\mathbf{B}_{q+1,\eta}(f)}\frac{|\mathbb{G}_n[\mathbf{R}_{t_\ast}]^T\boldsymbol{\theta}|}{\sqrt{\boldsymbol{\theta}^T\mathbb{V}_n(\mathbf{R}_{t_\ast})\boldsymbol{\theta}}}\right]}
<1+b\right)-(1-\mathbb{P}_P(B_{V}\leq \nu))\\
&\geq& 
1-\frac{1}{1+b}-(1-\mathbb{P}_P(B_{V}\leq \nu)),
\end{eqnarray*}
where the last inequality is due to Markov's inequality. 
Therefore, the statement of the theorem follows. 
\end{proof}

\section{Identification and Estimation from the Previous Literature}

This appendix section presents the identification and estimation for the characteristic function $\varphi_{X}$ and the density function $f_{X}$ of $X$ based on \cite{li/vuong:1998} and its extensions.
Moreover, a choice of the tuning parameter based on \citep{delaigle/gijbels:2004} is also reviewed.
Although the main text of this paper is focused on inference, one would also want to present estimates along with confidence bands as we presented in Figure \ref{fig:instances}.
This appendix section provides a method of obtaining estimates for convenience of readers.

\subsection{Identification and Estimation of the Characteristic Functions}\label{sec:li_vuong}
For a joint distribution $P$ of $(Y_1,Y_2)$, \citet{li/vuong:1998} show that the characteristic functions of $X$ and $U_1$ are identified by
\begin{align}
\varphi_X(t) &= \exp\left( \int_0^t \frac{i \mathbb{E}_P\left[Y_1 e^{i\tau Y_2}\right]}{\mathbb{E}_P\left[e^{i\tau Y_2}\right]} d\tau \right)
\qquad\text{and}
\label{eq:li_vuong}\\
\varphi_{U_1}(t) &= \frac{\mathbb{E}_P\left[e^{itY_1}\right]}{\exp\left( \int_0^t \frac{i \mathbb{E}_P\left[Y_1 e^{i\tau Y_2}\right]}{\mathbb{E}_P\left[e^{i\tau Y_2}\right]} d\tau \right)}
\notag
\end{align}
respectively, under the assumption of nonvanishing characteristic function of $Y_2$ in addition to Assumption \ref{a:continuous_independence}.
The sample-counterpart estimator of (\ref{eq:li_vuong}) reads
\begin{align}
\widehat\varphi_X(t) &= \exp\left( \int_0^t \frac{i \mathbb{E}_n\left[Y_1 e^{i\tau Y_2}\right]}{\mathbb{E}_n\left[e^{i\tau Y_2}\right]} d\tau \right).
\label{eq:li_vuong_estimator}
\end{align}
Similarly,
$$
\widehat\varphi_{U_1}(t) = \frac{\mathbb{E}_n\left[e^{itY_1}\right]}{\exp\left( \int_0^t \frac{i \mathbb{E}_n\left[Y_1 e^{i\tau Y_2}\right]}{\mathbb{E}_n\left[e^{i\tau Y_2}\right]} d\tau \right)}
\notag.
$$

\subsection{Derivatives of the Characteristic Function}\label{sec:derivative}
The identification and estimation method of \citet{li/vuong:1998} can be extended to the derivatives of the characteristic functions.
Specifically, taking the log derivatives of (\ref{eq:li_vuong}):
$$
\lambda_{X}(t)=\log\varphi_X(t),
$$
we can obtain
$$
\lambda_{X}^{(1)}(t)=\frac{i\mathbb{E}_P\left[Y_1 e^{it Y_2}\right]}{\varphi_{Y_2}(t)}
$$
and hence
$$
\varphi_{Y_2}(t)\lambda_{X}^{(1)}(t)=i\mathbb{E}_P\left[Y_1 e^{it Y_2}\right].
$$
Taking up to the fourth-order derivatives, we obtain
\begin{align*}
\varphi_{Y_2}(t)\lambda_{X}^{(2)}(t)+\varphi_{Y_2}^{(1)}(t)\lambda_{X}^{(1)}(t)
&=
-\mathbb{E}_P\left[Y_1Y_2 e^{it Y_2}\right]
\\
\varphi_{Y_2}(t)\lambda_{X}^{(3)}(t)
+2\varphi_{Y_2}^{(1)}(t)\lambda_{X}^{(2)}(t)
+\varphi_{Y_2}^{(2)}(t)\lambda_{X}^{(1)}(t)
&=
-i\mathbb{E}_P\left[Y_1Y_2^2 e^{it Y_2}\right]
\\
\varphi_{Y_2}(t)\lambda_{X}^{(4)}(t)
+3\varphi_{Y_2}^{(1)}(t)\lambda_{X}^{(3)}(t)
+3\varphi_{Y_2}^{(2)}(t)\lambda_{X}^{(2)}(t)
+\varphi_{Y_2}^{(3)}(t)\lambda_{X}^{(1)}(t)
&=\mathbb{E}_P\left[Y_1Y_2^3 e^{it Y_2}\right]
\end{align*}
We write the above equations as the linear equation:
\begin{align*}
\left(
\begin{array}{cccc}
\varphi_{Y_2}(t)&0&0&0\\
\varphi_{Y_2}^{(1)}(t)&\varphi_{Y_2}(t)&0&0\\
\varphi_{Y_2}^{(2)}(t)&2\varphi_{Y_2}^{(1)}(t)&\varphi_{Y_2}(t)&0\\
\varphi_{Y_2}^{(3)}(t)&3\varphi_{Y_2}^{(2)}(t)&3\varphi_{Y_2}^{(1)}(t)&\varphi_{Y_2}(t)
\end{array}
\right)
\left(
\begin{array}{c}
\lambda_{X}^{(1)}(t)\\
\lambda_{X}^{(2)}(t)\\
\lambda_{X}^{(3)}(t)\\
\lambda_{X}^{(4)}(t)
\end{array}
\right)
&=
\left(
\begin{array}{c}
i\mathbb{E}_P\left[Y_1 e^{it Y_2}\right]\\
-\mathbb{E}_P\left[Y_1Y_2 e^{it Y_2}\right]\\
-i\mathbb{E}_P\left[Y_1Y_2^2 e^{it Y_2}\right]\\
\mathbb{E}_P\left[Y_1Y_2^3 e^{it Y_2}\right]
\end{array}
\right).
\end{align*}
Assuming the invertibility of the $4 \times 4$ matrix on the left-hand side, we explicitly write $(\lambda_{X}^{(1)}(t),\dots,\lambda_{X}^{(4)}(t))'$ in terms of observable moments as
\begin{align*}
\left(
\begin{array}{c}
\lambda_{X}^{(1)}(t)\\
\lambda_{X}^{(2)}(t)\\
\lambda_{X}^{(3)}(t)\\
\lambda_{X}^{(4)}(t)
\end{array}
\right)
&=
\left(
\begin{array}{cccc}
\mathbb{E}_P\left[e^{itY_2}\right]&0&0&0\\
i\mathbb{E}_P\left[Y_2e^{itY_2}\right]&\mathbb{E}_P\left[e^{itY_2}\right]&0&0\\
-\mathbb{E}_P\left[Y_2^2e^{itY_2}\right]&2i\mathbb{E}_P\left[Y_2e^{itY_2}\right]&\mathbb{E}_P\left[e^{itY_2}\right]&0\\
-i\mathbb{E}_P\left[Y_2^3e^{itY_2}\right]&-3\mathbb{E}_P\left[Y_2^2e^{itY_2}\right]&3i\mathbb{E}_P\left[Y_2e^{itY_2}\right]&\mathbb{E}_P\left[e^{itY_2}\right]
\end{array}
\right)^{-1}
\left(
\begin{array}{c}
i\mathbb{E}_P\left[Y_1 e^{it Y_2}\right]\\
-\mathbb{E}_P\left[Y_1Y_2 e^{it Y_2}\right]\\
-i\mathbb{E}_P\left[Y_1Y_2^2 e^{it Y_2}\right]\\
\mathbb{E}_P\left[Y_1Y_2^3 e^{it Y_2}\right]
\end{array}
\right).
\end{align*}

With these $\varphi_{X}(t),\lambda_{X}^{(1)}(t),\dots,\lambda_{X}^{(4)}(t)$ written explicitly written in terms of observable moments, we can in turn identify $\varphi_X^{(1)}(t),\varphi_X^{(2)}(t),\varphi_X^{(3)}(t),\varphi_X^{(4)}(t)$ in terms of observable moments as follows:
$$
\varphi_X^{(1)}(t)=\varphi_X(t)\lambda_{X}^{(1)}(t),
$$
\begin{align*}
\varphi_X^{(2)}(t)
&=
\varphi_X^{(1)}(t)\lambda_{X}^{(1)}(t)+\varphi_X(t)\lambda_{X}^{(2)}(t)\\
&=
\varphi_X(t)(\lambda_{X}^{(1)}(t)^2+\lambda_{X}^{(2)}(t)),
\end{align*}
\begin{align*}
\varphi_X^{(3)}(t)
&=
\varphi_X^{(1)}(t)(\lambda_{X}^{(1)}(t)^2+\lambda_{X}^{(2)}(t))+\varphi_X(t)(2\lambda_{X}^{(1)}(t)\lambda_{X}^{(2)}(t)+\lambda_{X}^{(3)}(t))\\
&=
\varphi_X(t)\left(\lambda_{X}^{(1)}(t)^3+3\lambda_{X}^{(1)}(t)\lambda_{X}^{(2)}(t)+\lambda_{X}^{(3)}(t)\right)
\end{align*}
\begin{align*}
\varphi_X^{(4)}(t)
&=
\varphi_X^{(1)}(t)\left(\lambda_{X}^{(1)}(t)^3+3\lambda_{X}^{(1)}(t)\lambda_{X}^{(2)}(t)+\lambda_{X}^{(3)}(t)\right)
\\&+
\varphi_X(t)\left(3\lambda_{X}^{(1)}(t)^2\lambda_{X}^{(2)}(t)
+3\lambda_{X}^{(1)}(t)\lambda_{X}^{(3)}(t)
+3\lambda_{X}^{(2)}(t)^2
+\lambda_{X}^{(4)}(t)\right)\\
&=
\varphi_X(t)
\left(
\lambda_{X}^{(1)}(t)^4
+6\lambda_{X}^{(1)}(t)^2\lambda_{X}^{(2)}(t)
+4\lambda_{X}^{(1)}(t)\lambda_{X}^{(3)}(t)
+3\lambda_{X}^{(2)}(t)^2
+\lambda_{X}^{(4)}(t)\right)
\end{align*}
Sample counterparts of these derivatives, along with (\ref{eq:li_vuong_estimator}), can be used to estimate $M$.

\subsection{Tuning Parameter}\label{sec:tuning_parameters_li_vuong}
To estimate the probability density function $f_{X}$ of $X$ using the characteristic function estimator (\ref{eq:li_vuong}), we need to impose a regularization by limiting the integration for the Fourier transform to a contact interval $[-h^{-1},h^{-1}]$ for some ``bandwidth'' $h$.
Finite-sample choice methods of choosing the limit frequency $h$ are proposed in the literature of deconvolution kernel density estimation.
One of the most widely used approaches is to minimize the MISE \citep{stefanski/carroll:1990} or its asymptotically dominating part \citep{delaigle/gijbels:2004}:
\begin{equation*}
AMISE(h) = \frac{1}{2\pi nh} \int \left| \frac{\phi_K(t)}{\varphi_{U_1}(t/h)} \right|^2 dt +\frac{h^4}{4}  \int u^2 K(u) du\cdot\int f^{(2)}_{X}(x)^2 dx.
\end{equation*}
where $\varphi_K$, supported on $[-1,1]$, is $\mathcal{F}K$ for some kernel function $K$.

There are alternative ways to compute $\int f^{(2)}_{X}(x)^2 dx$.
Based on Parseval's identity, \citet{delaigle/gijbels:2004} suggest 
\begin{equation*}
\int f^{(2)}_{X}(x)^2 dx = \frac{1}{2\pi h^5} \int t^4 \frac{|\varphi_X(t/h)|^2 |\varphi_K(t)|^2}{|\varphi_{U_1}(t/h)|^2} dt.
\end{equation*}
Combining the above two equations together yields
\begin{equation*}
AMISE(h) = \frac{1}{2\pi nh} \int \left| \frac{\phi_K(t)}{\varphi_{U_1}(t/h)} \right|^2 dt +\frac{1}{8\pi h}  \int u^2 K(u) du \cdot \int t^4 \frac{|\varphi_X(t/h)|^2 |\varphi_K(t)|^2}{|\varphi_{U_1}(t/h)|^2} dt.
\end{equation*}
With this formula, one may choose $h$ to minimize the plug-in counterpart of $AMISE(h)$, replacing the unknown characteristic functions $\varphi_X$ and $\varphi_{U_1}$ by the sample counterparts $\widehat\varphi_X$ and $\widehat\varphi_{U_1}$, respectively, in Appendix \ref{sec:li_vuong}.

Since the set $[-h^{-1}, h^{-1}]$ of frequencies is used for estimation, it is also a natural idea to use this set $[-h^{-1}, h^{-1}]$ of frequencies for inference as well, although our theory for inference does not require such a finite limit unlike the estimation which requires regularization. 

\subsection{Estimation of the Density Function}\label{sec:density}
With the estimated characteristic function (\ref{eq:li_vuong_estimator}) and the bandwidth parameter $h$ chosen in Section \ref{sec:tuning_parameters_li_vuong}, the density function may be estimated by
$$
\widehat f_X(x) = \frac{1}{2\pi} \int e^{-itx} \varphi_K(th) \widehat\varphi_X(t) dt.
$$
The ``Li-Vuong estimates'' shown in Section \ref{sec:simulation} are based on the above formula together with the tuning parameter chosen according to the procedure outlined in Appendix \ref{sec:tuning_parameters_li_vuong}.

\section{Additional Proofs}
\subsection{Proof of Proposition \ref{prop:hermite_assumption} (Sufficient Condition for Assumption \ref{a:orthonormal})}\label{sec:prop:hermite_assumption}
\begin{proof}
First, it follows from \citet[][Theorem 16.3.1]{blanchard/bruening:2002} that $\Psi=\{\psi_j:j=0,1,\ldots\}$ satisfies Assumption \ref{a:orthonormal} (i).
Furthermore, since $|\psi_j| \leq 1.086435\pi^{-1/4}$ for each $j=0,1,\ldots$ \citep[see e.g.,][p. 208]{Erdelyi/etal:1953}, we have
$$
|\langle{ f, \psi_j }\rangle|
\leq
\langle{ f(x), |\psi_j| }\rangle
\leq
1.086435\pi^{-1/4} \cdot \int f(x)dx
=
1.086435\pi^{-1/4}
$$
for each $f \in \mathcal{L}_0(P)$
for each $P \in \mathcal{P}$ and 
for each $j=0,1,\ldots$.
This shows that Assumption \ref{a:orthonormal} (ii) is satisfied with $\Theta^{q+1} = \left[-1.086435\pi^{-1/4},1.086435\pi^{-1/4}\right]^q$.
\end{proof}

\subsection{Proof of Proposition \ref{prop:bounds} (Approximation Bounds)}\label{sec:prop:bounds}
\begin{proof}
First, note that the Hermite function $\psi_j$ in (\ref{eq:hermite_function}) satisfies the Schr\"odinger equation:
\begin{equation}\label{eq:schrodinger}
\psi_j^{(2)}(x) = -(2j+1-x^2) \cdot \psi_j(x)
\end{equation}
for each $j=0,1,\ldots$ \citep[][Theorem 6.14 (6.41)]{folland:2009}.
Second, note that the Hermite functions (\ref{eq:hermite_function}) also satisfy the recurrence relation:
\begin{equation}\label{eq:recurrence}
\psi_j^{(1)} = \sqrt{\frac{j}{2}} \psi_{j-1} - \sqrt{\frac{j+1}{2}} \psi_{j+1}
\end{equation}
for each $j=1,2,\ldots$ \citep[][Theorem 6.14 (6.39)--(6.40)]{folland:2009}.
We will use these properties of the Hermite functions in the proof below.

By Proposition \ref{prop:hermite_assumption}, we can write
$$
f=\sum_{j=0}^{\infty}\langle{f,\psi_j}\rangle\psi_j.
$$
-- see \citet[][Theorem 5.27]{folland2007}.
Taking the second derivatives of the both sides, we obtain
\begin{eqnarray*}
f^{(2)}(x)
=
\sum_{j=0}^{\infty}\langle{f,\psi_j}\rangle\psi_j^{(2)}(x)
=
-\sum_{j=0}^{\infty}\langle{f,\psi_j}\rangle(2j+1-x^2)\psi_j(x),
\end{eqnarray*}
where the second equality is due to (\ref{eq:schrodinger}).
Rearranging, we have
\begin{eqnarray*}
f^{(2)}(x)+x^2f(x)
=
-\sum_{j=0}^{\infty}\langle{f,\psi_j}\rangle(2j+1)\psi_j(x).
\end{eqnarray*}
Further taking the second derivatives of the both sides yields
\begin{align*}
\frac{d^2}{d x^2} (f^{(2)}(x)+x^2f(x))
&=
-\sum_{j=0}^{\infty}\langle{f,\psi_j}\rangle(2j+1)\psi_j^{(2)}(x)\\
&=
\sum_{j=0}^{\infty}\langle{f,\psi_j}\rangle(2j+1)(2j+1-x^2)\psi_j(x),
\end{align*}
where the second equality is again due to (\ref{eq:schrodinger}).
Rearranging terms, we obtain
\begin{equation}\label{eq:hermite_fourth_derivative}
\frac{d^2}{d x^2} (f^{(2)}(x)+x^2f(x)) + x^2 \cdot (f^{(2)}(x)+x^2f(x))
=
\sum_{j=0}^{\infty}\langle{f,\psi_j}\rangle(2j+1)^2\psi_j(x)
\end{equation}

Combining (\ref{eq:hermite_smoothness}) and (\ref{eq:hermite_fourth_derivative}) together, we have $\left\| \sum_{j=0}^{\infty}\langle{f,\psi_j}\rangle(2j+1)^2\psi_j( \cdot ) \right\|_2^2\leq M$,
and hence
$$
\sum_{j=0}^{\infty}(\langle{f,\psi_j}\rangle)^2(2j+1)^4\leq M.
$$
Define $D=\sum_{j=q+1}^{\infty}(2j+1)^{-3}$.
The above inequality implies
$$
D\sum_{j=q+1}^{\infty}\dfrac{(2j+1)^{-3}}{D}(\langle{f,\psi_j}\rangle)^2(2j+1)^{7}\leq M.
$$
Using Jensen's inequality, we can write 
$$
\sqrt{D}\sum_{j=q+1}^{\infty}\dfrac{(2j+1)^{-3}}{D}\sqrt{(\langle{f,\psi_j}\rangle)^2(2j+1)^{7}}\leq\sqrt{D}\sqrt{\sum_{j=q+1}^{\infty}\dfrac{(2j+1)^{-3}}{D}(\langle{f,\psi_j}\rangle)^2(2j+1)^{7}}\leq \sqrt{M},
$$
which implies 
\begin{equation}\label{coeff_decay}
\sum_{j=q+1}^{\infty}\sqrt{2j+1}|\langle{f,\psi_j}\rangle|\leq \sqrt{MD}.
\end{equation}
Thus, we obtain
\begin{align*}
\sup_{f\in\mathcal{L}}\sup_{x\in I}\left|\sum_{j=q+1}^{\infty}\langle{ f,\psi_j }\rangle \cdot \psi_j(x)\right|
&\leq
\left(\sup_{f\in\mathcal{L}}\sum_{j=q+1}^{\infty}\sqrt{2j+1}\left|\langle{ f,\psi_j }\rangle\right|\right) \cdot \sup_{j=q+1,\ldots}\frac{\sup_{x\in I}\left|\psi_j(x)\right|}{\sqrt{2j+1}}\\
&\leq
\frac{1.086435\pi^{-1/4}}{\sqrt{2q+3}}\sqrt{MD}.
\end{align*}

Next, we note that Hermite function is the eigenfunction of the Fourier transform operator.
Specifically, $|\phi_j| =\sqrt{2\pi}|\psi_j|\leq 1.086435\pi^{-1/4}\sqrt{2\pi}$ holds.
Thus, similar lines of calculations to those above yield
\begin{align*}
\sup_{f\in\mathcal{L}}\sup_{t\in \mathbb{R}}\left|\sum_{j=q+1}^{\infty}\langle{ f,\psi_j }\rangle \cdot \phi_j(t)\right|
&\leq
\left(\sup_{f\in\mathcal{L}}\sum_{j=q+1}^{\infty}\sqrt{2j+1}\left|\langle{ f,\psi_j }\rangle\right|\right) \cdot \sup_{j=q+1,\ldots}\frac{\sup_{x\in I}\left|\phi_j(x)\right|}{\sqrt{2j+1}}\\
&\leq
\frac{1.086435\pi^{-1/4}\sqrt{2\pi}}{\sqrt{2q+3}}\sqrt{MD}
\end{align*}

Finally, if $q \in \mathbb{N}$, then we also obtain
\begin{eqnarray*}
&&
\sup_{f\in\mathcal{L}}\left|\sum_{j=q+1}^{\infty} \langle{ f,\psi_j }\rangle \cdot \left(i\phi_j(t)\cdot\mathbb{E}_P\left[Y_1\exp(itY_2)\right]-\phi_j^{(1)}(t)\cdot\mathbb{E}_P\left[\exp(itY_2)\right]\right)\right| \\
&\leq&
\sup_{f\in\mathcal{L}}\sum_{j=q+1}^{\infty}\left|\langle{ f,\psi_j }\rangle \right|\cdot \left(|\phi_j(t)|\cdot\mathbb{E}_P\left[|Y_1|\right]+|\phi_j^{(1)}(t)|\right) \\
&\leq&
\sup_{f\in\mathcal{L}}\sum_{j=q+1}^{\infty} \sqrt{2j+1}\left|\langle{ f,\psi_j }\rangle \right|\cdot \frac{|\phi_j(t)|\cdot\mathbb{E}_P\left[|Y_1|\right]+\sqrt{j/2}|\phi_{j-1}(t)|+\sqrt{(j+1)/2}|\phi_{j+1}(t)|}{\sqrt{2j+1}}\\
&\leq&
\sup_{f\in\mathcal{L}}\sum_{j=q+1}^{\infty}\sqrt{2j+1} \left|\langle{ f,\psi_j }\rangle \right|\cdot 
\sup_{j=q+1,\ldots}\frac{|\phi_j(t)|\cdot\mathbb{E}_P\left[|Y_1|\right]+\sqrt{j/2}|\phi_{j-1}(t)|+\sqrt{(j+1)/2}|\phi_{j+1}(t)|}{\sqrt{2j+1}}\\
&\leq&
\sqrt{MD}\left(\frac{\mathbb{E}_P\left[|Y_1|\right]}{\sqrt{2q+3}}+1\right)\cdot\frac{1.086435\pi^{-1/4}}{\sqrt{2\pi}}
%
%
%
\end{eqnarray*}
where the second inequality is due to (\ref{eq:recurrence}).
This completes a proof of the proposition.
\end{proof}

\subsection{Proof of Proposition \ref{prop:equivalent_smoothness} (Equivalent Smoothness Condition)}\label{sec:prop:equivalent_smoothness}
\begin{proof}
By the Parseval's identity, we can write
$
\int \Xi(x)^2 dx
=
\int \left| \mathcal{F}\Xi(t) \right|^2 dt
$
for $\Xi \in \mathcal{L}^2$.
Therefore, we compute $\int \left| \mathcal{F}\Xi(t) \right|^2 dt$, where
$\Xi(x) = \frac{d^2}{d x^2} (f''(x)+x^2f(x)) + x^2 \cdot (f''(x)+x^2f(x))$.

First, note that we can rewrite
\begin{equation*}
\Xi(x) = f^{(4)}(x) + 2x^2 f^{(2)}(x) + (x^2+4x)f^{(1)}(x) + (x^4+2)f(x).
\end{equation*}
To evaluate $\mathcal{F} \Xi$, we evaluate each of the four terms in the right-hand side as follows:
\begin{align*}
\int e^{itx} f^{(4)}(x) dx
&= t^4 \int e^{itx}f(x) dx,
\\
\int e^{itx} \left( 2x^2 f^{(2)}(x) \right) dx
&= -2t^2 \int x^2 e^{itx}f(x) dx + 8it \int x e^{itx}f(x) dx + 4 \int e^{itx}f(x) dx,
\\
\int e^{itx} \left( (x^2+4x)f^{(1)}(x) \right) dx
&= -it \int x^2 e^{itx}f(x) dx- (4it + 2) \int x e^{itx}f(x) dx - 4 \int e^{itx}f(x) dx,
\\
\int e^{itx} \left( (x^4+2)f(x) \right) dx
&= \int x^4 e^{itx}f(x) dx + 2 \int e^{itx}f(x) dx,
\end{align*}
where the first three equalitie follows through integration by parts under $\lim_{|x| \rightarrow \infty} x^2 f(x) = \lim_{|x| \rightarrow \infty} x^2 f^{(1)}(x) = \lim_{|x| \rightarrow \infty} f^{(2)}(x) = \lim_{|x| \rightarrow \infty} f^{(3)}(x) = 0$. 
Adding these four equations together, we obtain
\begin{align*}
\mathcal{F}\Xi(t) = \int x^4 e^{itx}f(x) dx - (2t^2 + it) \int x^2 e^{itx}f(x) dx -(2 - 4it) \int x e^{itx}f(x) dx + (t^4 + 2) \int e^{itx}f(x) dx.
\end{align*}
Therefore, the statement of the proposition follows.
\end{proof}

\subsection{Derivation of Eq. (\ref{eq:assn2_sufficient}) (A Sufficient Condition for Assumption \ref{a:size_control} (ii))}\label{appendix:proof_assn2_sufficient}
\begin{proof}
Note that 
\begin{eqnarray*}
\left|(\mathbf{R}_{t}-\mathbb{E}_P[\mathbf{R}_{t}])^T\boldsymbol{\theta}\right|
&\leq&
\sup_{\psi=\psi_0,\ldots,\psi_q}\left|{R}_{\psi,t}(Y_1,Y_2)-\mathbb{E}_P[{R}_{\psi,t}(Y_1,Y_2)]\right|\cdot\left\|\boldsymbol{\theta}\right\|\\
&\leq&
\sup_{\psi=\psi_0,\ldots,\psi_q}\left(2|Y_1|+2\mathbb{E}_P[|Y_1|]+4|{\phi}^{(1)}(t)|\right)\cdot\left\|\boldsymbol{\theta}\right\|\\
&=&
\left(2|Y_1|+2\mathbb{E}_P[|Y_1|]+4\sup_{\psi=\psi_0,\ldots,\psi_q}|{\phi}^{(1)}(t)|\right)\cdot\left\|\boldsymbol{\theta}\right\|
\end{eqnarray*}
and 
\begin{eqnarray*}
\left|(\mathbf{I}_{t}-\mathbb{E}_P[\mathbf{I}_{t}])^T\boldsymbol{\theta}\right|
&\leq&
\left(2|Y_1|+2\mathbb{E}_P[|Y_1|]+4\sup_{\psi=\psi_0,\ldots,\psi_q}|{\phi}^{(1)}(t)|\right)\cdot\left\|\boldsymbol{\theta}\right\|,
\end{eqnarray*}
because 
\begin{eqnarray*}
|{R}_{\psi,t}(Y_1,Y_2)|
&=&
\left|-\cos(tY_2)(Y_1\mathrm{Im}({\phi}(t))+\mathrm{Re}({\phi}^{(1)}(t)))-\sin(tY_2)(Y_1\mathrm{Re}({\phi}(t))-\mathrm{Im}({\phi}^{(1)}(t)))\right|\\
&\leq&
|Y_1|+|\mathrm{Re}({\phi}^{(1)}(t))|+|Y_1|+|\mathrm{Im}({\phi}^{(1)}(t))|\\
&\leq&
2|Y_1|+2|{\phi}^{(1)}(t)|.
\end{eqnarray*}
Also note that, for the Hermite functions, we have 
\begin{eqnarray*}
|\psi_j^{(1)}|
&=&
|i^j\sqrt{2\pi}\psi_j^{(1)}|
\\
&=&
\sqrt{2\pi}|\psi_j^{(1)}|
\\
&=&
\sqrt{2\pi}\left|\sqrt{j/2}\psi_{j-1}^{(1)}-\sqrt{(j+1)/2}\psi_{j+1}^{(1)}\right|
\\
&\leq&
\sqrt{2\pi}(\sqrt{j/2}|\psi_{j-1}^{(1)}|+\sqrt{(j+1)/2}|\psi_{j+1}^{(1)}|)
\\
&\leq&
\sqrt{2\pi}(\sqrt{j/2}+\sqrt{(j+1)/2})\times 1.086435\pi^{-1/4}
\\
&\leq&
4\sqrt{j+1}. 
\end{eqnarray*}
Since $\boldsymbol{\theta}^T\mathbb{V}_P(\mathbf{R}_{t_l})\boldsymbol{\theta}\geq\|\boldsymbol{\theta}\|^2\mathrm{eig}_{\min}(\mathbb{V}_P(\mathbf{R}_{t}))$ and $\boldsymbol{\theta}^T\mathbb{V}_P(\mathbf{I}_{t_l})\boldsymbol{\theta}\geq\|\boldsymbol{\theta}\|^2\mathrm{eig}_{\min}(\mathbb{V}_P(\mathbf{I}_{t}))$, 
we have 
\begin{eqnarray*}
M_{L,q,k}(\boldsymbol{\theta},P)
&\leq&
\dfrac{
\max_{1 \leq l \leq L} \max
\left\{\mathbb{E}_P\left[\left|(\mathbf{R}_{t_l}-\mathbb{E}_P[\mathbf{R}_{t_l}])^T\boldsymbol{\theta}\right|^k\right]^{1/k}, \mathbb{E}_P\left[\left|(\mathbf{I}_{t_l}-\mathbb{E}_P[\mathbf{I}_{t_l}])^T\boldsymbol{\theta}\right|^k\right]^{1/k} \right\}
}
{
\sqrt{
\min_{1 \leq l \leq L}\min\left\{
\boldsymbol{\theta}^T\mathbb{V}_P(\mathbf{R}_{t_l})\boldsymbol{\theta},
\boldsymbol{\theta}^T\mathbb{V}_P(\mathbf{I}_{t_l})\boldsymbol{\theta}
\right\}
}
}
\\
&\leq&
\dfrac{
4\mathbb{E}_P\left[|Y_1|^k\right]^{1/k}+4\mathbb{E}_P[|Y_1|]+8\sup_{\psi=\psi_0,\ldots,\psi_q}|{\phi}^{(1)}(t)|
}
{
\sqrt{
\min_{t\in[-T,T]}\min\left\{
\mathrm{eig}_{\min}(\mathbb{V}_P(\mathbf{R}_{t})),
\mathrm{eig}_{\min}(\mathbb{V}_P(\mathbf{I}_{t}))
\right\}
}
}\\
&\leq&
\dfrac{
4\mathbb{E}_P\left[|Y_1|^k\right]^{1/k}+4\mathbb{E}_P[|Y_1|]+32\sqrt{q+1}
}
{
\sqrt{
\min_{t\in[-T,T]}\min\left\{
\mathrm{eig}_{\min}(\mathbb{V}_P(\mathbf{R}_{t})),
\mathrm{eig}_{\min}(\mathbb{V}_P(\mathbf{I}_{t}))
\right\}
}
}\\
\end{eqnarray*}
and
\begin{eqnarray*}
B_{L,q}(\boldsymbol{\theta},P) 
&\leq& 
\dfrac{
\mathbb{E}_P\left[\max_{1 \leq l \leq L} \max
\left\{\left|(\mathbf{R}_{t_l}-\mathbb{E}_P[\mathbf{R}_{t_l}])^T\boldsymbol{\theta}\right|^4 , \left|(\mathbf{I}_{t_l}-\mathbb{E}_P[\mathbf{I}_{t_l}])^T\boldsymbol{\theta}\right|^4 \right\} \right]^{1/4}
}
{
\sqrt{
\min_{1 \leq l \leq L}\min\left\{
\boldsymbol{\theta}^T\mathbb{V}_P(\mathbf{R}_{t_l})\boldsymbol{\theta},
\boldsymbol{\theta}^T\mathbb{V}_P(\mathbf{I}_{t_l})\boldsymbol{\theta}
\right\}
}
}
\\
&\leq& 
\dfrac{
4\mathbb{E}_P\left[|Y_1|^4\right]^{1/4}+4\mathbb{E}_P[|Y_1|]+8\sup_{\psi=\psi_0,\ldots,\psi_q}|{\phi}^{(1)}(t)|
}
{\sqrt{
\min_{t\in[-T,T]}\min\left\{
\mathrm{eig}_{\min}(\mathbb{V}_P(\mathbf{R}_{t})),
\mathrm{eig}_{\min}(\mathbb{V}_P(\mathbf{I}_{t}))
\right\}
}
}\\
&\leq& 
\dfrac{
4\mathbb{E}_P\left[|Y_1|^4\right]^{1/4}+4\mathbb{E}_P[|Y_1|]+32\sqrt{q+1}
}
{\sqrt{
\min_{t\in[-T,T]}\min\left\{
\mathrm{eig}_{\min}(\mathbb{V}_P(\mathbf{R}_{t})),
\mathrm{eig}_{\min}(\mathbb{V}_P(\mathbf{I}_{t}))
\right\}
}
}\\
\end{eqnarray*}
Since $\mathbb{E}_P\left[Y_1^4\right]<\infty$, 
we have 
\begin{eqnarray*}
&&\left(M_{L,q,3}^3(\boldsymbol{\theta},P) \vee M_{L,q,4}^2(\boldsymbol{\theta},P) \vee B_{L,q}(\boldsymbol{\theta},P) \right)^2
\\
&&=
O\left(\left(\frac{\sqrt{q+1}
}
{\sqrt{
\min_{t\in[-T,T]}\min\left\{
\mathrm{eig}_{\min}(\mathbb{V}_P(\mathbf{R}_{t})),
\mathrm{eig}_{\min}(\mathbb{V}_P(\mathbf{I}_{t}))
\right\}
}
}\right)^6\right)
\\
&&=
O(n^{1/2-c_1}\log^{-7/2}(4Ln)).
\end{eqnarray*}
\end{proof}

\end{document}